\documentclass{article}
\usepackage{times}
\usepackage{amsfonts,amssymb,amsmath,amsthm}
\usepackage{graphicx}
\def\mpfile#1#2{\includegraphics{#1#2.eps}}
\def\mpfilescaled#1#2#3{\includegraphics[scale=#3]{#1#2.eps}}

\title{Orbits of linear maps and regular languages}
\author{S. Tarasov\thanks{Supported in part by RFBR grant
    08--01--00414.}\\
serge99meister@gmail.com
\and
M. Vyalyi\thanks{Supported  by
RFBR grant 09--01--00709 and the scientific school grant %
NSh5294.2008.1.}\\ 
vyalyi@gmail.com
 }
\date{December 5, 2010}
\newtheorem{theorem}{Theorem}
\newtheorem{lemma}{Lemma}
\newtheorem{prop}{Proposition}

\newtheorem{cor}{Corollary}
\theoremstyle{definition}

\newtheorem{remark}{Remark} 
\let\ph\varphi
\let\eps\varepsilon
\let\al\alpha
\def\Tb{T^\text{\textup{bad}}}
\def\Ta#1{T^{\text{\textup{all},\:}#1}}
\def\BB{\mathbb B}

\def\ZZ{\mathbb Z}
\def\NN{\mathbb N}

\def\QQ{\mathbb Q}
\def\fI{\ensuremath{I_\BB}}
\def\fP{\ensuremath{P_\BB}}
\def\fS{\ensuremath{S_\BB}}
\def\fP{\ensuremath{P_\BB}}
\def\id{\mathrm{id}}
\def\reg(#1){$#1$-reali\-za\-bi\-li\-ty}
\def\Per{\ensuremath{\mathrm{Per}}}
\def\PePe{(\Per_{\Sigma_1}{\parallel}P_{\Sigma_2})}
\def\pepe{\Per_{\Sigma_1}{\parallel}P_{\Sigma_2}}
\begin{document}

\maketitle

\begin{abstract} We settle the equivalence between the problem of
hitting a polyhedral set by the orbit of a linear map and the 
intersection of a regular language and a language of permutations of
binary words  
(\fP-realizability problem).

The decidability of the both problems is presently unknown and the first one 
is a straightforward generalization of the famous Skolem problem and the  
nonnegativity problem in the theory of 
linear recurrent sequences.

To show a `borderline' status of 
\fP-realizability problem with respect to computability we
present some decidable and undecidable problems
closely related to it.

This paper is an extended version of the journal publication
\cite{VyaTar10} and contains some additional results.
\end{abstract}

\section*{Introduction}

Let $\Phi$
be a linear map of a vector space $V$ into itself
and let 
$x\in V$ 
be a vector 
fin $V$.  
The iterations of $\Phi$
applied to $x$ 
define
an \emph{orbit} 
$\mathop{\mathrm{Orb}}\nolimits_\Phi{}x$, 
i.e. the set
\begin{equation*}
\{\Phi^kx: k\in{\mathbb Z}^+\}. 
\end{equation*}

In the present paper we discuss algorithmic issues related to orbits.
We assume that $V$ is a rational coordinate space, so $\Phi$ and $x$ are 
represented by their (rational) components.

\emph{An orbit description problem} consists in finding specific
relations which either hold for all vectors in the orbit or are
violated by at least one vector in the orbit. Here we limit ourself to
a simple case where the relations are formed from Boolean combinations
of linear equalities and inequalities (see the exact definition below
in 
Section~\ref{Orb&LRS}).

Note that the most important case of the orbit problem is 
\emph{the chamber hitting problem\/}. In this particular case we check
whether an orbit intersects a closed polyhedron (i.e. a solution set for 
a finite system of nonstrict linear inequalities).

The orbit description problems are related to some 
problems on
linear recurrent sequences.

\emph{A linear recurrent sequence}
(LRS) $x_n$ of a degree $d$
is defined by:
\begin{equation}\label{LRSdef}
  \left\{\begin{aligned}
    &x_n=\sum_{i=1}^d a_{i}x_{n-i}&&\ \text{when }n>d,\\ 
    &x_n=b_n&&\
    \text{when } 1\leq n\leq d,\ \text{where }a_i,b_j\ \text{are constants.}
  \end{aligned}
  \right.
\end{equation}

The famous Skolem problem is perhaps the most known algorithmic 
problem on
linear recurrences.

\medskip

\textbf{The Skolem problem.}
Let $\{x_n\}$ 
be a specified  LRS  with integer coefficients. Whether $x_k=0$ for some $k$?

\medskip

The decidability of the Skolem problem is presently an open question, but 
it is known that it is decidable for the degrees 
$\leq5$ (the cases $d=3, 4$ are worked out by N.~Vereshchagin~\cite{Ver},
and the case $d=5$ is solved by
V.\,Halava~et~al.~\cite{Halava05}). 

In the opposite direction, the best hardness result on the Skolem problem is
NP-hardness~\cite{BP02}.

LRS is called \emph{nonnegative\/} if all its elements are
nonnegative.  One more important algorithmic question about LRS is
called the \emph{nonnegativity problem\/} and consists in checking
nonnegativity of a LRS. This problem is at least as hard as the Skolem
problem.  Being a bit more formal it means that the Skolem problem is
Turing reducible to the nonnegativity problem. This statement
certainly belongs to the public mind, but see the proof of the
statement in 
Section~\ref{reduction1}.

The nonnegativity problem is decidable for LRS of the degree~$\leq3$
(V.~Laohakosol, P.~Tangsupphathawat~\cite{LT09}). 
 
In this paper we often use some standard constructions from the theory of
algorithms and the complexity theory, e.g., Turing reducibility,
$m$-reducibility, polynomial reducibility. In particular, Turing
reducibility of a {\em Problem I} to {\em Problem II} means that 
while solving {\em Problem I}
the
reduction may ask an oracle, who can give answers to {\em Problem II}
(see, further~\cite{ShV,GJ}.)

The relations between the orbit description problems and LRS 
are described in Section~\ref{Orb&LRS}. 

These results are certainly not new but we include them for
completeness sake and in order to introduce necessary notation and
terminology.

Let briefly sketch the contents of the paper.

We recall basic properties of  LRS in Section~\ref{alg-comb}. 
LRS are closely related to regular languages. In particular, 
any LRS can be represented as a {\em difference} of the generating
functions of a pair of regular languages (see, e.g., \cite[Cor.~8.2]{SaSo}.
We will use a  modification of this result (see
Theorem~\ref{LR->2Auto} from  Section~\ref{HCP}). 
(All necessary facts about regular languages could be found
in~\cite{HopcroftMotwaniUllman,BeRe}.)   

The main result of the present paper consists in an 
algorithmic equivalence
between the  chamber hitting problem
and checking a particular property of the regular languages. Namely,
the property involved consists in checking whether a regular language
contains at least one word from a special set of words. We call this
set a \emph{permutation filter\/}.  Speaking informally, an arbitrary
word from the permutation filter gives a permutation of all binary words
of a fixed length~$n$. 
See a formal statement in Section~\ref{Filter} below.

The permutation filter is denoted by
$P_{\mathbb B}$ and the corresponding problem of checking this property
of regular languages is called
$P_{\mathbb B}$-realizability problem (see, Section~\ref{Filter}).

The proof of an algorithmic equivalence of a problem of
$P_{\mathbb B}$-realizability and the chamber hitting problem
(Theorem~\ref{PermF=Skolem} from Section~\ref{Filter}) 
proceeds
in
several steps. 

At first, 
we describe in Section~\ref{HCP} a polynomial
reducibility of the chamber hitting problem to the $P_{\mathbb
B}$-realizability problem. 
The reducibility uses the aforementioned Theorem~\ref{LR->2Auto}.

As a consequence of this result we prove (see,
theorem~\ref{NPhardness}) NP-hardness of a 
$P_{\mathbb B}$-realizability problem. This result may be of
independent interest. 

To construct a reduction in the opposite direction, i.e. a reduction
of the $P_{\mathbb B}$-realizability problem to the chamber hitting
problem, we settle some technical difficulties.
It turns out that a natural
construction described in Subsection~\ref{route}
gives a reduction of the $P_{\mathbb B}$-realizability problem to the
problem of hitting 
a translate of an integral polyhedral cone represented by generators. 
To reduce this problem to the 
problem of hitting a rational cone we use some additional technical 
tricks and their description is contained in Subsection~\ref{Zhitsubsect}.  
As an intermediate step 
we consider the case of a simplicial cone 
(the generating vectors are linear independent).
The general case is reduced to the simplicial one using
representation of an arbitrary integral cone as a union of a finite
set and a finite family of translates of simplicial integral cones
(Theorem~\ref{FrobNdim}). 
This result may be regarded as a sort of an integral Caratheodory's
theorem.
Recall that integral
variants of the Caratheodory's theorem are actively studied (see, for
instance,~\cite{CookW, EisenShmonin})
and have important applications to combinatorial optimization.

Finally, in Section~\ref{decid-and-undec} we
present some decidable and undecidable problems
closely related to \fP-realizability problem thus demonstrating its
`borderline' status with respect to computability.

\section{Algebraic and combinatorial properties of LRS}
\label{alg-comb}

Below we remind some  algebraic and combinatorial properties of LRS
which will be used 
in the sequel. All needed proofs and references 
may be found in~\cite{Stanley,BeRe,Halava05}.

1. Let $A$ and $h$ be, respectively, a linear operator 
and a linear function on a vector space 
$V$ and let  $x,y\in V$. Then the generating function
\begin{equation}\label{lin-func}
  f_{A,x,y}(t)=\sum_{r\geq0} h(A^kx-y)t^r
\end{equation}
is rational.

2. The generating function of any LRS is rational and Taylor's series
expansion coefficients of an arbitrary rational function at any point
in its domain form a LRS.

3. The set of LRS is closed under componentwise addition and
multiplication (i.e. 
under Hadamard product of sequences).

In the standard setting, the input to 
the Skolem problem is an integer LRS.
But if we are interested in signs of the elements of LRS  involved  only, then
there is no difference between integral and rational cases.
 
Indeed, let a LRS be specified by rational data and let
$N$ be the LCM of the denominators of all $a_i$ and $b_i$.
Form an integral LRS
\begin{equation}\label{ZLRS-def}
  y_n=\sum_{i=1}^d N^{i}a_{i}y_{n-i}\ \text{if }n>d,\quad y_n=N^{n+1}b_n\ 
  \text{if } 1\leq n\leq d.  
\end{equation}

\begin{prop}\label{Q->Z}
 $N^{n+1}x_n=y_n$ for all $n$.
\end{prop}
\begin{proof}
  For $1\leq n\leq{}d$ the statement follows from the
  definition~\eqref{ZLRS-def}. 
  For  $n>d$ 
  we proceed by induction.
  Assuming that the statement holds for
  all $n< k$,
  we get
  \begin{equation*}
    y_{k}=\sum_{i=1}^d N^{i}a_{i}y_{k-i}=
    \sum_{i=1}^d N^{i}a_{i}(x_{k-i}N^{k-i+1})=
    \sum_{i=1}^dN^{k+1}a_{i}x_{k-i}=N^{k+1}x_{k}.
  \end{equation*}
  Thus the statement holds for
  $n=k$.
\end{proof}

\begin{remark}
Trivially, using the Euclidean algorithm, the identity
$$
\text{LCM}(x,y)=\frac{xy}{\text{LCF}(x,y)}
$$
and by Proposition~\ref{Q->Z}  we can compute $y_n$ from $x_n$ in time 
polynomial in the length of the input data.
So
for algorithmic problems, 
which are concerned with  the signs of LRS  elements only,
integral and rational cases are equivalent w.r.t. polynomial
reductions.

Hereinafter we assume that in algorithmic problems an LRS is represented by
the list of coefficients $a_i$, $b_j$ written in binary.
\end{remark}

Some LRS are solutions of enumeration  problems and for
this reason are nonnegative. For instance, the family of the so called
${\mathbb N}$-rational sequences coincide with the set of generating
functions of regular languages (see,~\cite{BeRe}).  
In the sequel
${\mathbb N}$-rational sequences are called
\emph{regular sequences\/}.

To be more precise, any deterministic automaton
$A$ over the alphabet
$\Sigma$
defines LRS $s_n$, where
\begin{equation}\label{auto-seq}
  s_n(A)=\#\{w: \mbox{ word } w \mbox{ is accepted by } A\mbox{ and } |w|=n\}.
\end{equation}
Exactly the sequences of the type~\eqref{auto-seq} will be  called regular.

LRS are not regular as they may have negative elements, but nevertheless
the following statement holds.

\begin{theorem}[{\cite[Cor.~8.2]{SaSo}}]\label{diff-th}
Any LRS is a difference of two regular sequences.  
\end{theorem}

\section{Orbits of linear maps and LRS}\label{Orb&LRS}
In this section we state formally the orbit description problem and 
the related chamber hitting problem, and indicate relations 
of these problems to LRS.

Let $h_1,\dots,h_m$ be a family of affine
functions defined on a coordinate space ${\mathbb Q}^d$.  The
sign patterns of these functions induce a partition of ${\mathbb Q}^d$
into chambers.  Formally, let $s\in\{\pm1,0\}^m$.  Then \emph{a
chamber} is a set
\begin{equation*}
  \{x\in{\mathbb Q}^d : \mathop{\mathrm{sign}}(h_i(x))=s_i\ \text{for } 1\leq i\leq m\},
\end{equation*}
where $\mathop{\mathrm{sign}}(t)$
is a standard sign function
\begin{equation*}
  \mathop{\mathrm{sign}}(t)=\left\{
  \begin{aligned}
    1,&&&\text{for } t>0,\\
    0,&&&\text{for } t=0,\\
    -1,&&&\text{for } t<0.\\
  \end{aligned}
  \right.
\end{equation*}

\medskip

\textbf{Orbit description problem} (ODP).

{\sc INPUT:}
a square matrix $\Phi$ of order $d$; 
$d$-dimensional vector $x_0$; 
a family of affine functions $h_1,\dots,h_m$ 
on ${\mathbb Q}^d$ and a set of sign patterns
$s_1,\dots, s_r\in\{\pm1,0\}^m$.

{\sc OUTPUT:}
`yes' if  any point of the orbit
$\mathop{\mathrm{Orb}}\nolimits_\Phi{}x_0$ 
falls into the union of chambers
$H_{s_1}\cup,\dots, \cup H_{s_r}$, and `no' otherwise.

\begin{remark}
  We assume that matrices, vectors and affine functions in the ODP and
  all related problems are represented by the component lists, where
  the components are written in binary. 
\end{remark}

\textbf{Chamber hitting problem} (CHP) has the same input as the ODP,
but there is one sign pattern $s$ only.

The output of CHP is `yes' if the orbit
$\mathop{\mathrm{Orb}}\nolimits_\Phi{}x_0$ 
intersects the chamber $H_s$ and `no' otherwise.

\begin{prop}\label{CHP=ODP}
  \textup{CHP} is Turing-equivalent to \textup{ODP}. 
\end{prop}
\begin{proof}
  To reduce the CHP to the ODP we take the  complement in the set
$\{\pm1,0\}^m$ to the sign pattern of the chamber in the input 
of an instance of the CHP and fix all other input parameters. We
obtain an instance of the ODP that outputs `yes' iff the instance of
the CHP outputs `no'.

The reduction in the opposite direction is proved analogously.
For any chamber not included in the input list of an instance of the
ODP we solve the corresponding CHP.  All the answers are `no' iff
the instance of the ODP reports `yes'.
\end{proof}

\begin{remark}
  It is rather obvious that the CHP is recursevely enumerable (belongs
  to the~class $\Sigma^1$ of the arithmetic hierachy) and the ODP is
  co-enumerable (belongs to the class $\Pi^1$). The proof of
  Proposition~\ref{CHP=ODP} shows that these problems are
  complementary in 
  a broad
  sense.

  Another pair of
  complementary problems is the chamber description problem (the ODP restricted
  to the case of one chamber) and the problem of hitting the complement
  to a chamber. 

  We do not know reducibilities between the chamber description
  problem and the CHP.
\end{remark}

\begin{prop}\label{Skolem=Hit1}
The nonnegativity problem is equivalent to the \textup{ODP} 
restricted to  one linear function and a nonstrict inequality (the chambers
$H_{0}$,  $H_{+1}$).

The Skolem problem is equivalent to the \textup{CHP} 
restricted to one linear function and equality (the chamber
$H_{0}$).
\end{prop}
\begin{proof}
  As is well known the LRS~\eqref{LRSdef} 
  is related to a linear operator~$A$ in the 
  $d$-dimen\-sional coordinate space given  in matrix notation by
  \begin{equation}
    A=\begin{pmatrix}
      a_{1}&a_{2}&\dots&a_{d-1}&a_d\\
      1&     0&\dots     &   0&0\\
      0& 1&     \dots   &     0&0\\
      \hdotsfor{5}\\
      0&0&\dots&1&0
    \end{pmatrix}
  \end{equation}

  It is easy to check by induction that
$$
A^k(b_{d},\dots,b_1)^T = (x_{k+d},x_{k+d-1},\dots,x_{k+1})^T.
$$

Thus, the Skolem problem for the  LRS~\eqref{LRSdef} 
is reduced to the CHP with
$\Phi=A$, $x_0=(b_{d},\dots,b_1)^T$, $h_1=x_d$ 
and the chamber
$H_0$.

Analogously, the nonnegativity problem is reduced to the ODP with
the same input as above and two chambers
$H_0$, $H_{+1}$.

The reductions in the opposite direction use standard facts about LRS
listed in section~\ref{alg-comb}.  Indeed, the generating function
$\sum_n h(A^nx_0)t^n$ is rational.  Hence, the sequence $h(A^nx_0)$ is
an LRS and both reductions follow.
\end{proof}

We don't know whether the general ODP and CHP can be reduced to,
respectively, the nonnegativity problem and the Skolem problem, but we
can indicate particular cases when such reductions do exist.

\medskip

\textbf{Union of subspaces hitting problem} (SHP).

{\sc INPUT:}
a square matrix $\Phi$ of order $d$; 
a $d$-dimensional vector $x_0$; 
a family of linear functions $h_{jk}$ 
on ${\mathbb Q}^d$,
$1\leq j\leq m$, $1\leq k\leq
 r_j$. 

{\sc OUTPUT:}
`yes' if  the orbit
$\mathop{\mathrm{Orb}}\nolimits_\Phi{}x_0$ 
intersects the  union of affine subspaces
 \begin{equation*}
   \bigcup_{j=1}^m\{x: h_{j1}(x)=h_{j2}(x)=\dots=h_{jr_j}(x)=0\}
 \end{equation*}
and `no' otherwise.

\medskip

\begin{lemma}\label{Qhit->Skolem}
  The \textup{SHP} 
is reducible to the Skolem problem.
\end{lemma}
\begin{proof}
It follows from properties of LRS listed in Section~\ref{alg-comb} 
that the sequences
  $$
  \varphi_n(j,k)=h_{jk}(\Phi^nx_0),\, j=0,\dots, m,
  $$
are LRS. But as LRS are closed under componentwise sum and product,
the sequence 
  $$
  \varphi_n=\prod_{j=0}^m \sum_{k=1}^{s_j} \varphi_n(j,k)^2
  $$
is also LRS.

Note that taking the input of the SHP we can algorithmically compute
the representation of~$\varphi_n$ in the form~\eqref{LRSdef}.  And the
reduction of the SHP to the Skolem problem follows as $\varphi_n=0$
iff for some $j$ it holds
  $$
  \varphi_n(j,1)=\varphi_n(j,2)=\dots=\varphi_n(j,s_j)=0.
  $$
\end{proof}

\begin{remark}
Note that $m$-reducibility in the proof above 
may not be a polynomial one as the degree of the sequence
$\varphi$ 
may be exponential w.r.t. the initial degree $d$.
But the answer to a general SHP is disjunction of answers to separate
subspace hitting problems. And for this particular case the
reducibility described in the proof is polynomial.  Hence, we conclude
that the SHP can be solved in polynomial time with the Skolem
problem oracle. In other words, the SHP is polynomial time Turing
reducible to the Skolem problem.
\end{remark}

The famous Skolem-Mahler-Lech theorem~\cite{ BeRe, Halava05,Tao} 
asserts that the set of zeroes 
of a LRS consists of a finite set and a
union of a finite number of arithmetical progressions. 

The proof of
Lemma~\ref{Qhit->Skolem} implies that the same is true for the SHP.

\begin{cor}[\bf Skolem-Mahler-Lech theorem for SHP]
The set of those $k$ for which $\Phi^kx_0$ falls into the union of
subspaces $V_1,\dots,V_{m}$, is a union of a finite set and a a finite
number of arithmetical progressions.
\end{cor}

\medskip

\textbf{Polyhedron localization problem} (PLP)
is a particular case of ODP when the family of chambers forms a closed
convex
polyhedron (i.e. a solution set of a system of nonstrict linear
inequalities).

\medskip

\begin{lemma}\label{Qcover->nonneg}
  The \textup{PLP} 
  is reducible to the nonnegativity problem.
\end{lemma}

\begin{proof}
W.l.o.g. we assume that the sign pattern of the interior of a polyhedron
is $(+1,\dots,+1)$ 
and  sign patterns for the faces of the polyhedron are obtained by
replacing some $1$'s  by $0$'s. 
Indeed, all negative components in a
sign pattern can be removed by a sign change. And all zero components
in the sign pattern of  the interior of a polyhedron
can be removed after checking that the orbit falls into (affine)
subspace resulting from the solution of a system of linear equations
(corresponding to zero components of the sign pattern). The last
problem is certainly decidable as the affine hull of the orbit
coincides with the affine hull of the first $d+1$ orbit points.

Now rearrange the sequences $\varphi_n(j)=h_j(\Phi^nx_0)$ into one
sequence $\varphi_n$ by consequently putting the elements of
$\varphi_n(j)$ into places with an index having residue $j$
modulo
  $m$, i.e. $\varphi_{sm+j}\stackrel{def}{=}\varphi_s(j),\, s\geq 0,
0< j<m$. 
The resulting sequence $\varphi_n$ is an LRS. Indeed, its generating 
function  $f(t)$ is rational as it satisfies the identity
 $$
  f(t)=\sum_{j=1}^m t^j f_j(t^m),
  $$
where 
for any $j=1,\dots,m$, a function $f_j(t)$  is a generating function for
$\varphi_n(j)$. 
 Hence, $f(t)$ is a finite sum of rational functions. 
\end{proof}

The Skolem problem is the weakest of all problems mentioned above.  To
decide the Skolem problem it is enough to enumerate the family of LRS
with nonnegative elements or the family of the PLP instances with
`yes' answers (due to $m$-reducibility in Lemma~\ref{Qcover->nonneg}).

The following proposition belongs to folklore.

\begin{prop}\label{reduction1}
If the family of LRS with nonnegative coefficients 
is recursively enumerable then 
the
Skolem problem is decidable.
\end{prop}
\begin{proof}
  If a sequence $x_n$ is a LRS then the sequence $x_n^2-1$ is also a LRS as
  the family of all LRS is closed under componentwise sum/difference and
product, see, Section~\ref{alg-comb}.  But the sequence $x_n^2-1$ has
nonnegative elements iff $x_n\ne0$ for all~$n$. 
To complete the proof we apply E.~Post theorem~\cite{ShV}. We
enumerate all LRS with nonnegative coefficients and compare them to
$x_n^2-1$.  In parallel we enumerate all LRS having zeroes and compare
them to $x_n^2-1$ also. One of these enumerations should eventually
stop.
\end{proof}

\section{CHP and regular languages}\label{Filter} 

In this section we relate
the
Skolem problem to some properties of regular
languages. To be more exact we define problems of \emph{regular
realizability\/}: whether a given regular language contains a word of
a particular kind.  It is convenient to describe the problem of
realizability as follows.  Let $L$ be a language in a finite alphabet
$\Sigma$.  Informally, $L$ encodes a {\em property} that is checked.
The input to the problem of $L$-realizability is a description of some
regular language~$R\subseteq\Sigma^*$ and we are asked whether the
intersection $L\cap R$ is nonempty.

\begin{remark}
Depending on a way of presentation of a regular language we obtain
different in general variants of $L$-realizability problem.  We assume
that $R$ is given via deterministic automaton accepting it. These
clarification is of course immaterial if we are interested in
decidability only, but may be important for complexity
estimates.
\end{remark}

\emph{The permutation filter} $P_{\mathbb B}$ 
is a language over the alphabet $\{\#,0,1\}$ of all words
$\#w_1\#{}w_2\#{}\dots w_N\#$, where $N=2^n,\, 
n\geq1$, 
and the set $\{w_i\},\, i=1,2,\dots,N$ 
consists of all  binary words of the length~$n$. 

Words from $P_{\mathbb B}$ are called \emph{permutation words\/}.
Informally, a permutaton word corresponds to a permutation
of a set of all binary words of some fixed length~$n$. This length is
called the \emph{block rank} of a permutation word.

\begin{theorem}\label{PermF=Skolem}
\textup{CHP}  and $P_{\mathbb B}$-realizability problem are Turing equivalent.
\end{theorem}

To prove the theorem we construct reductions in both directions, but
as a matter of fact, their complexities are essentially different. 

In the next section we polynomially reduce the CHP to the
$P_{\mathbb B}$-realizability problem and obtain as a corollary
$NP$-hardness of the last problem.  This follows from the fact that 
the Skolem problem is $NP$-hard~\cite{BP02} and according to
proposition~\ref{Skolem=Hit1} 
%5-12
it 
is polynomially reducible to CHP. 

\begin{remark}
  A different proof of $NP$-hardness of the $P_{\mathbb
  B}$-realizability problem can be derived from the properties of periodic
  filters, see~\cite{Vya09}.
\end{remark}

The reduction in the opposite direction is shown 
in Section~\ref{Filter->Qhit} but it takes superexponential time.

\section{Polynomial reduction of the CHP to\\
the  $P_{\mathbb B}$-realizability problem}\label{HCP} 

We briefly sketch the idea of the reduction.  First,
using~\eqref{lin-func} we associate a LRS with any linear constraint
(equality or inequality) in the description of a chamber.  Then by
proposition~\ref{Q->Z} we convert this LRS into an integral one.
According to Theorem~\ref{diff-th} this LRS can be represented as a
difference of two regular sequences. It turns out that this
construction is not enough to establish polynomial reduction and we
elaborate it slightly and represent the integral LRS involved as a
difference of two regular sequences on some explicit arithmetic
progression of its indices. The last representation can be computed in
polynomial time.  Next we find a special automaton that compares
numbers of words of specified length in two regular languages provided
the input is a permutation word. This trick gives a reduction of the Skolem
problem to the $P_{\mathbb B}$-realizability problem. To finish the proof
we need an additional construction converting several automata
described above into one.

We proceed to the proof.

At first note that for any square matrix $A$ of order $d$,
$d$-dimensional vector $x_0$ and a linear function $h$ having rational
coefficients there is a polynomial time algorithm to recover the
coefficients of LRS for the sequence $h(A^nx_0)$.  Indeed, the degree
of the LRS involved does not exceed the order of the matrix and,
hence, the coefficients of LRS and the initial data can be determined
if we solve the corresponding system of linear equations for the first
$2d$ elements of the sequence.

In this section  we assume that all LRS involved have integer coefficients
and integral 
initial data. This is possible due to
Proposition~\ref{Q->Z} as we are interested in signs of the expressions  
$h(A^nx_0)$ only. 

Next we express a LRS as a sum of weights of the walks in a digraph
with fixed start and finish vertices.  Let $\Gamma(V,E)$ be a digraph
(loops and parallel edges allowed) and let $c\colon E\to{\mathbb Z}$
be a weight function for edges of $\Gamma$.  A~weight of a walk is a
product of weights over the edges of the walk.

\begin{lemma}\label{LR->GW}
  There exists a polynomial time algorithm that takes an integer LRS
  $\{x_n\}$ and outputs a weighted digraph $G_x$ with two marked
  vertices $s$ and $f$, such that the sum of weights over all walks of
  the length~$n$ starting at the vertex~$s$ and finishing at the
  vertex~$f$ equals~$x_n$ for all~$n\geq1$.
\end{lemma}
\begin{proof} Take a LRS~\eqref{LRSdef} and construct a digraph
$G_x$ as follows. (See Fig.~\ref{LRgraph}.)

\begin{figure}
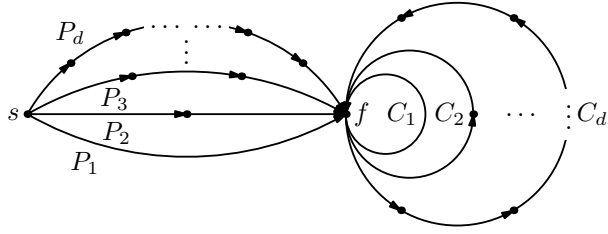

  \centerline{\mpfile{ga}{1}}
  \caption{A graph for a LRS}\label{LRgraph}
\end{figure}

  Connect the start vertex $s$ and the finish vertex $f$ by $d$
  vertex-disjoint 
(except for the terminals)
paths $P_1$, \dots, $P_d$ such that 
the edge
length of
  $P_i$ is~$i$. Add $d$ vertex-disjoint 
(except for the terminal $f$) 
directed
cycles $C_1$, \dots, $C_d$,
  passing through the vertex~$f$. 
The 
edge
length of~$C_i$ is also~$i$.
  The cycles $C_i$ have no common vertices with 
the
paths $P_j$ except for
  the vertex~$f$.

Now set edge weighting of $G_x$.
All edges of $G_x$ except for the first edges of all
paths $P_i$ and cycles $C_j$ 
has weight~$1$.
Set the weight of the first (outgoing from $s$) edge  on the path $P_i$ 
to $p_i,\, i=1,\dots,d$,
and set the  weight of the 
first (outgoing from $f$) edge on the cycle $C_i$ to $q_i,\, i=1,\dots,d$.
  
  Let~$z_n$ be the sum of weights over all walks of the length~$n$
  from~$s$ to $f$. There is a recurrence on~$z_n$.
  Note that for $n\leq d$ there exists a walk of the length~$n$ from~$s$ to
  $f$ along a
  path of the digraph and for $n>d$ any walk of the length $n$ is
  obtained from a shorter one by adding a cycle. Thus
  \begin{align*}
    z_i &= p_i +\sum_{j=1}^{i-1}q_{i-j} z_{j}, &&\text{for } 1\leq
    i\leq d;\\
    z_n &=\sum_{j=1}^{d}q_{j} z_{n-j}, &&\text{for } n>d.
  \end{align*}
  To ensure equalities $z_n=x_{n}$ 
  set
  \begin{equation*}
    q_{j}=a_j.
  \end{equation*}
  For $1\leq i\leq d$ equalities $z_i=b_i$ imply
  \begin{align*}
    p_1 &= b_1;\\
    p_2 &= b_2 - q_{1}b_1;\\
    p_3 &= b_3 - q_{1}b_2- q_{2}b_1;\\
    &\dots
  \end{align*}

  To complete the proof note that all computations above  take a
  polynomial time.
\end{proof}

\begin{lemma}\label{Graph-numerical}
  There exists a polynomial time algorithm that takes positive
   integers~$n$ (in binary) and $k$ (in unary) and outputs a digraph,
   a start vertex~$s$ and a finish vertex~$f$ such that the number of
   paths of the length~$k$ from~$s$ to $f$ equals~$n$ provided
   $k>\log_2n$.
\end{lemma}
\begin{proof}
  Let  $L_t$ be a doubled directed path having $t+1$ vertices
  $\{1,2,\dots,t+1\}$ and 
pairs of parallel edges $(i,i+1)$ between any consecutive
  vertices $i$ and $i+1,\, i=1,\dots,t$.  Clearly, there are exactly $2^t$
  walks of the length~$t$ between vertices $1$ and $t+1$.

A \emph{thread\/} $L_j'$ is a digraph obtained from $L_j$ by attaching
a directed path of the length~$k-j$ to its finish vertex $j+1$.  The
start and the finish vertices of a thread are, respectively, the start
vertex of $L_j$ and the finish vertex of a path.

Let $\{j_1,\dots,j_l\} $ be the positions of nonzero bits in 
the binary representation~$n$. Take the threads $L_{j_i},\, i=1,\dots,l$ 
and identify their start vertices and, respectively, their
finish vertices into global start vertex $s$ and global finish vertex $f$. 
By construction there are exactly~$n$ paths of the length $k$ from~$s$ to~$f$ in
  the resulting graph.

  The case of $n=11$, $k=4$ is
  illustrated in Fig.~\ref{n-repr}.

\begin{figure}[h]
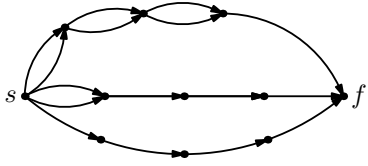

  \centerline{\mpfile{ga}{2}}
  \caption{Construction of Lemma~\ref{Graph-numerical} for $n=11$, $k=4$}\label{n-repr}
\end{figure}

Evidently the algorithm is polynomial.
\end{proof}

Now we are able to prove a variation of Theorem~\ref{diff-th}.

\begin{theorem}\label{LR->2Auto}
  There exists a polynomial time algorithm that 
  takes an integer LRS $x_n$ and outputs an integer~$\ell$ 
that 
is polynomially bounded in the
  input length (i.e. the total binary length of all coefficients and
  the initial data of the LRS)
and two
  deterministic automata $A$, $B$ over the alphabet~$\{0,1\}$ such
  that  $x_{n}=s_{\ell{}n}(A)-s_{\ell{}n}(B)$ for all~$n\geq1$. 
\end{theorem}

\begin{proof}
  The first step is described in the proof of Lemma~\ref{LR->GW}. 
  As a result we obtain a digraph~$H$. 

  Then we choose integers 
  $M>\log_2\max\{|a_1|,\dots,|a_{d}|,|b_1|,\dots,|b_{d}|\}$ and
  $k>\log_2(2Md)$.
  Now the integer $\ell$ to output is set to $M+k$.

  At the next step we apply Lemma~\ref{Graph-numerical} and transform
  the digraph~$H$ to a digraph satisfying two additiional properties:

1)  the fan-out of each vertex  
is~$2$;

2)  the 
  absolute value of each weight is~$1$.

  The transformation consists of the following steps.

  1. For each integer $a_i$ ($b_i$)
   apply Lemma~\ref{Graph-numerical} to construct  a digraph
  $A_i$ ($B_i$) 
   having exactly $|a_i|$ ($|b_i|$) paths of the length $M$ from
   the start  to the finish.

  2.  Replace the first edges of all cycles $C_i$ (respectively, of
 all paths $P_i$) of the graph~$H$ by pasting in the graphs $A_i$
 (respectively, $B_i$). Namely, the start vertex of a graph $A_i$
 ($B_i$) is identified with the initial vertex of an edge and the
 finish vertex is identified with the end vertex of the edge.
  
  Other edges of the graph $H$ are replaced by directed paths of
  the length~$\ell$.

  Fan-outs of all vertices in the resulting digraph $H'$ are at
  most~$2$ except for the start vertex $s$ and for the finish
  vertex~$f$. Note that fan-outs of the start and the finish vertices
  are at most $2Md$.

  3. Attach to the start and to the finish vertex of the graph~$H'$ 
  rooted directed binary trees of depth~$k$ having
  $2Md$~leaves. Replace the starting vertex of each edge outgoing the
  start vertex (the finish vertex) by a leaf of the corresponding
  tree in such a way that the fan-out of each vertex in the resulting graph
  $H''$ is at most~$2$.

\begin{figure}
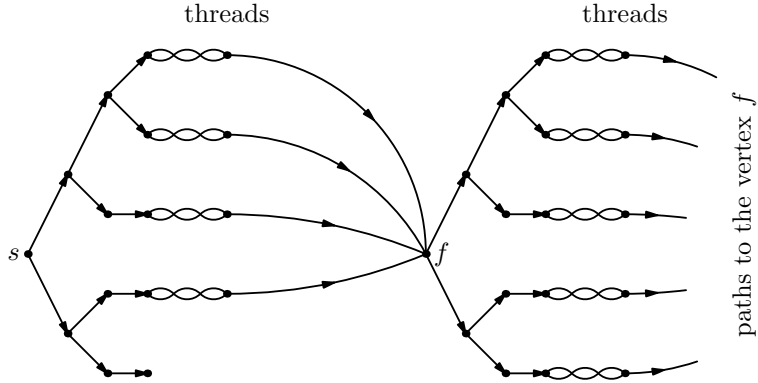

  \centerline{\mpfile{ga}{3}}
  \caption{A graph $H''$}\label{H'}
\end{figure}

  As a result we obtain a modified digraph~$H''$ (see Fig.~\ref{H'}). 
  It has vertices with
  the fan-out less than~$2$. Add an auxillary vertex $u$ and edges to the~$u$
  from all vertices with the fan-out less than~$2$ in the digraph~$H'$
  to make all fan-outs equal~$2$. Then add
  two loops at the vertex~$u$. After these operations we obtain a 
  digraph~$G$.

  Now we assign weights to the edges of the  digraph~$G$.
  The weight of an~edge ingoing to the start vertex of a thread of a
  digraph~$A_i$ ($B_i$) coincides with the sign of the corresponding
  integer $a_i$ ($b_i$). All other edges have  weights~$1$.

  It is clear from the construction that in the digraph~$G$ the sum of
  weights over all walks of the length $\ell{}n$ from $s$ to $f$
  equals~$x_{n}$.
  
  Now we construct  automata $A$ and $B$ over the alphabet
  $\{0,1\}$. For both of them the state set is $\{\pm1\}\times
  V(G)$. Choose a labelling of the edges of the
  digraph~$G$ by  0 and 1 such that for each vertex edge labels of the
  two outgoing edges are different. (Recall that the fan-out of any
  vertex of $G$ is~2.)

  Let's describe the automaton~$A$.

  The initial state of the~$A$ is $(+1,s)$ and the only accepting
  state is $(+1,f)$. Reading a symbol~$\alpha$ in the state $(\sigma,w)$
   the automaton~$A$ goes to the state~$(\sigma', w')$ if the edge
  $(w,w')$ in the digraph~$G$ has label $\alpha$. If the edge
  $(w,w')$  has positive
  weight then $\sigma'=\sigma$. Otherwise, $\sigma'=-\sigma$.
  The rule is illustrated in Fig.~\ref{AutP}.

\begin{figure}
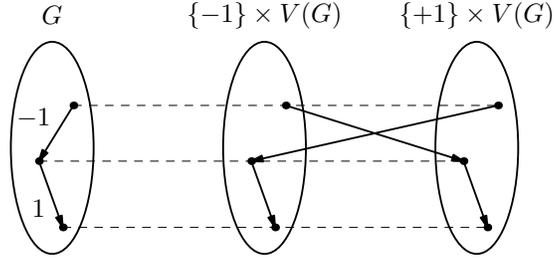
%[b]
  \centerline{\mpfile{ga}{4}}
  \caption{Counting the sign}\label{AutP}
\end{figure}

  The automaton~$B$ differs from the automaton~$A$ in the accepting
  state only. The accepting state of $B$ is~$(-1,f)$. 

  It is clear that the lengths of all words accepted by the
  automata~$A$ and $B$ are  multiples of~$\ell$.

  Note that both automata can be constructed from the initial LRS in
  polynomial time. 

  To finish the proof we note that $x_n$ equals the difference between
  the number of words of the length $\ell{}n$ accepted by the automata
  $A$ and $B$. Indeed, the walks of negative weight correspond to the
  words accepted by the~$B$ while the walks of positive weight
  correspond to the words accepted by the~$A$.
\end{proof}

Now we are ready to present a polynomial reduction of the Skolem
problem to the $P_{\mathbb B}$-realizability problem.

Theorem~\ref{LR->2Auto} implies that for the reduction it is sufficient
to build up an algorithm running in polynomial time such that takes a
pair of automata $A$, $B$ over the alphabet $\{0,1\}$ and an integer
$\ell$ in unary and outputs the description of an automaton~$C$ such
that  $L(C)\cap P_{\mathbb B}\ne\varnothing$ iff 
$s_{\ell{}n}(A)=s_{\ell{}n}(B)$ holds for some  $n$.

Let informally explain how to construct such an automaton.  
The automaton $C$ is made of two automata $C'$ and $C''$ by the
product construction. 
The automaton $C$ accepts iff $C' $ and $C''$ accept.  

Recall
that we are interested in operation of automata on permutation words,
i.e., words of the type $\#w_1\# w_2\#\dots w_k\#$, where
$w_i\in\{0,1\}^*,\, i=1,\dots,k$.  

The block rank (the length of
each $w_i$) should be a multiple of $l$.  This last condition is
verified by the automaton~$C''$. It counts the length of $w_i$ modulo
$l$ and rejects in the case of a nonzero residue. Otherwise it
accepts.

The automaton $C'$ starts from an accepting state $q_0$ and reads the
current block word $w_i$ separated by the delimiters $\#$.

If $w_i$ belongs to $L_{AB}=L(A)\setminus L(B)$ then $C'$ reads the next block
word $w_{i+1}$ and checks whether it belongs to $L_{BA}=L(B)\setminus
L(A)$. If the last check fails then $C'$ rejects, otherwise it returns
to the state $q_0$.

If $w_i$  belongs to $L_{BA}$ then $C'$ rejects.

If $w_i$  belongs to
$L_\sim=(L(A)\cap L(B)) \cup \overline{L(A)\cup L(B)}$ then $C'$ passes 
to a new accepting state $q_1$. 

Starting from the state $q_1$ the automaton $C'$ reads the
current block word~$w_i$. If the block belongs $L_{AB}\cup L_{BA}$
then $C'$ rejects. Otherwise it returns to the state $q_1$.

The structure of the automaton $C'$ is pictured in Fig.~\ref{Auto=}.

\begin{figure}
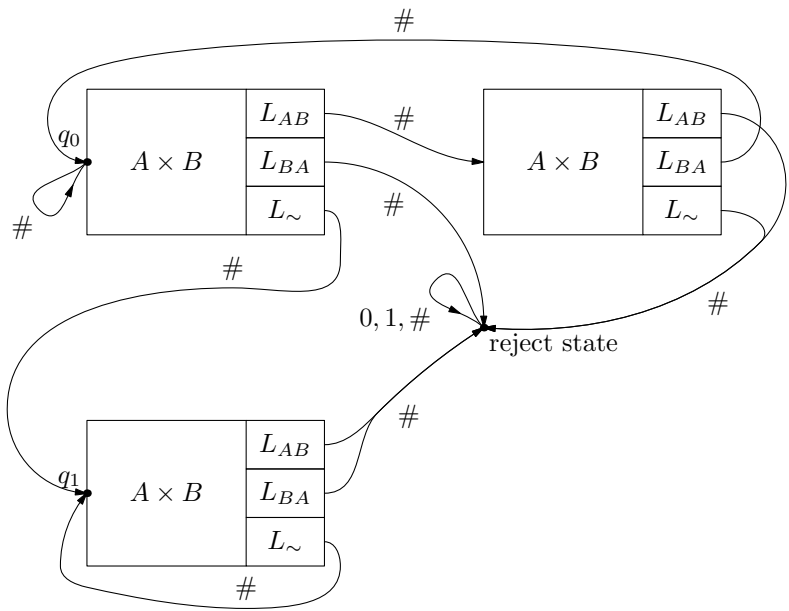

  \centerline{\mpfile{auto}{1}}
  \caption{The structure of the automaton $C'$. The initial state is
  $q_0$.  The only accepting states
  are $q_0$,~$q_1$}\label{Auto=}
\end{figure}

Let prove that $C$ gives the required reduction of the 
Skolem problem to the $P_{\mathbb B}$-realizability problem.

First we must show that if an instance of the Skolem problem has a
positive answer then $C$ accepts at least one word from
the permutation filter. 

Indeed, assume that some $x_n$ vanishes. Then
$s_{\ell{}n}(A)=s_{\ell{}n}(B)$ and there is a one-to-one
correspondence between the words of the length~$\ell{}n$ in $L_{AB}$
and the words of the length~$\ell{}n$ in $L_{BA}$.  A required
permutation word $W$ is obtained by arranging corresponding words in
pairs (other words of the length~$\ell{}n$ are arranged arbitrary
after all pairs).
By construction $C$ accepts~$W$.

On the other hand, assume that $C$ accepts some word $W$ from the
permutation filter. It follows from the construction that the block
rank of $W$ is a multiple of $\ell$. We can write $W$ in the form 
$W_1 W_2$. Here
$W_1$ is prefix of $W$ that is a concatenation of $\#w_i\# w_{i+1}$,
where $w_i\in L_{AB}$ and $w_{i+1}\in L_{BA}$.
And $W_2$ is a suffix of $W$ that is concatenation of $\#w_i$,
where $w_i\in L_\sim$. Thus,  $|L(A)\setminus L(B)|= |L(B)\setminus
L(A)|$. It means that $x_n=0$ and the instance of the Skolem problem
has a positive answer.

Clearly the size of the automata $C$ is polynomial in the size of the LRS.
This finishes the proof of the reduction.

From  the results of this section and the results from~\cite{BP02} we get
an immediate corollary.

\begin{theorem}\label{NPhardness}
  The problem of $P_{\mathbb B}$-realizability is \textup{NP}-hard.
\end{theorem}

Now we extend  the previous result to the general chamber hitting
problem. 

The construction of the automaton $C$ can be easily modified to accept
a permutation word satisfying the condition $s_{\ell{}n}(A)<
s_{\ell{}n}(B)$.  In the notation above to be accepted the suffix
$W_2$ of a permutation word should contain at least one more occurence
of a word from $L(B)\setminus L(A)$.

The counter part of the modified automaton $C_<$ does not change.
The structure of the automaton $C'_<$ is pictured in Fig.~\ref{Auto<}.

\begin{figure}
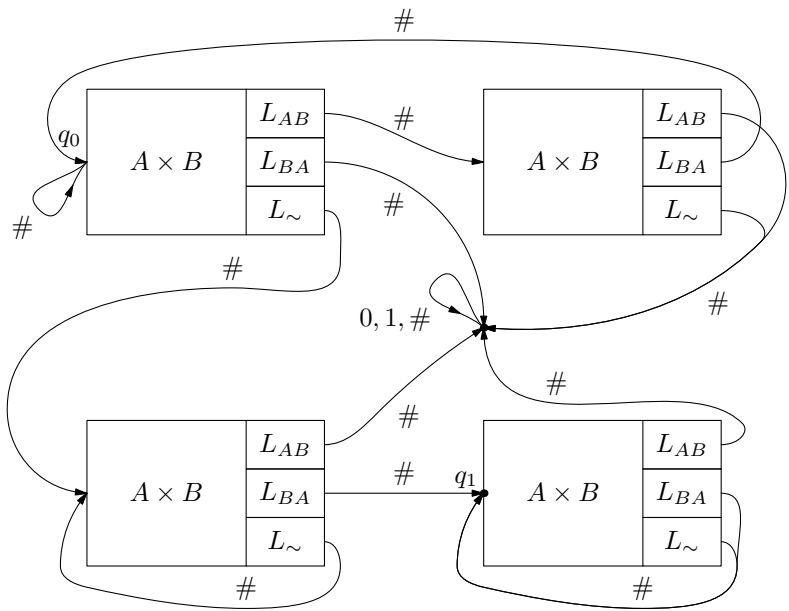

  \centerline{\mpfile{auto}{2}}
  \caption{The structure of the automaton $C'_<$. The initial state is
  $q_0$. The only accepting state is~$q_1$}\label{Auto<}
\end{figure}

It is evident that the size of the
modified automaton $C_<$ is upperbounded by the size of the automaton
$C$ up to a constant factor.

To complete a reduction of the chamber hitting problem to the
$P_{\mathbb B}$-realizability problem we construct an automaton that
 checks several conditions of the form\linebreak
$s_{\ell{}n}(A)=s_{\ell{}n}(B)$ ($s_{\ell{}n}(A)< s_{\ell{}n}(B)$).

Repeating the previous argument let's construct for each constraint
$h_i(\Phi^nx_0)<0$ (or $h_i(\Phi^nx_0)=0$) from the description of the
chamber a pair of automata $A_i$, $B_i$ and an integer $\ell$ common
to all pairs $A_i$, $B_i$ such that a constraint $h_i(\Phi^nx_0)<0$
($h_i(\Phi^nx_0)=0$) holds iff an inequality $s_{\ell{}n}(A)<
s_{\ell{}n}(B)$ (an equality $s_{\ell{}n}(A)= s_{\ell{}n}(B)$)
holds.
The main difficulty here stems from the condition that the block
ranks $\ell{}n$ 
of all permutation words certifying the corresponding constraints in the reducibility 
{\bf must be equal}. To solve this problem we use the following
arguments.

At first, choose the same integer $\ell$ for all pairs of
automata. The required value of $\ell$ is $1$ plus the three times
the maximum of all logarithms of all data for all LRS
involved. Then 
increase to $\ell/2$ the values of the parameters $k$ and $M$ from the proof
Theorem~\ref{LR->2Auto}.
It~can be done by attaching paths to threads of the graphs 
constructed in
Lemma~\ref{Graph-numerical}.

Next step is to construct for each pair $A_i$, $B_i$, where $1\leq
i\leq m$ and $m$ is the number of constraints of CHP, the automaton
$C'_i$ that certifies the~inequality $s_{\ell{}n}(A)< s_{\ell{}n}(B)$ (or
the equality $s_{\ell{}n}(A)= s_{\ell{}n}(B)$) as described above.

If $m$ is not a power of~$2$ then by definition an automaton $C'_i$,
where $m< i\leq 2^p$, $p=\lceil \log_2 m \rceil$, is a dummy 
automaton accepting all words. 

Now all automata $C'_i$ are combined into an automaton $C$ that checks
all conditions $s_{\ell{}n}(A)\, ?\, s_{\ell{}n}(B)$ for some length
$\ell{}n$.

\eject

The combined automaton has a product form $\tilde C=\tilde C'\times\tilde C''$
and accepts words $\#w_1\#{}w_2\#{}\dots w_N\#$ from the permutation
filter that satisfy the following requirements.
\begin{itemize}
\item [--] 
  The block rank is $p+\ell{}n$ for some $n$ (in the sequel we
  denote $w_i=u_iv_i$, where $|u|=p$).
\item [--] 
  Prefixes $u_i$ form a nondecreasing
  sequence w.r.t. the lexicographic order.
\item [--] For each $1\leq r\leq m $ let $z$ be the binary
  representation of $r-1$ of the length $p$. 
Form a subword  
  $
  \#w_i\# w_{i+1}\#\dots\# w_{i+j}\#
  $
  that consists of all block words with the prefix $z$.
  Then the word
  $\#v_{i}\#v_{i+1}\#\dots \#v_{i+j}\#$
  is  accepted by the automaton $C'_r$.
\end{itemize}

The first requirement is verified by a counter $\tilde C''$. The
second and the third requirements are verified by an automaton $\tilde
C'$, which is a sort of a concatenation of modified automata $\tilde
C'_i$.  In a modified automaton $\tilde C'_i$ each $\#$-transition
from a state $q'$ to a state $q''$ is changed by a $\#$-transition
from $q'$ to the initial state of a copy of an auxiliary automaton
$K_i$. This automaton reads a~word $z$ of the length $p$ and compares
it with $p$-length binary representations of $i-1$ and $i$. If $z$
represents $i-1$ then $K_i$ passes to the state $q''$ of the automaton
$C'_i$. If $z$ represents $i$ and $q''$ is an accepting state in the
automaton $C'_i$ then $K_i$ passes to the initial state of the
automaton $C'_{i+1}$. Otherwise it rejects.  This modification is
illustrated in Fig.~\ref{Auto_i}.

\begin{figure}
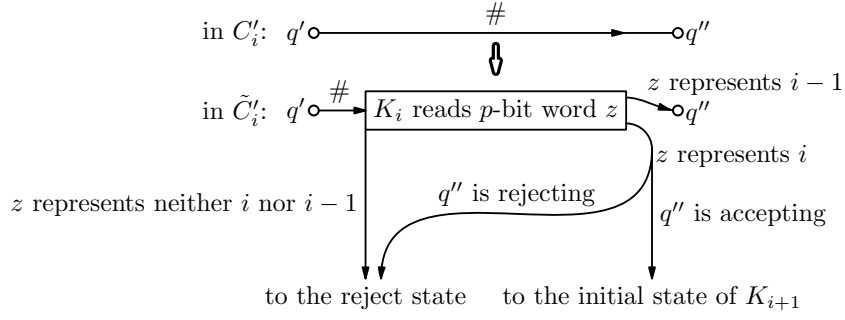

   \centerline{\mpfilescaled{auto}{3}{0.975}}
  \caption{Modification of a $\#$-transition in an automaton $C'_i$}
  \label{Auto_i}
\end{figure}

The initial state of the automaton $\tilde C'$ is the initial state of
$C'_1$. The accepting states are the initial states of those copies of the
auxiliary
automaton $K_{2^p}$ that are inserted instead of $\#$-transitions to 
accepting states of the automaton $C'_{2^p}$.

The size of the automaton $C$ is $O(Dm\log m )$, where $D$ is the
maximum of the sizes over all automata~$C_i$.  Indeed, an auxiliary
automaton $K_i$ can be implemented using $O(p)=O(\log m$) states. So
the size of a modified $C'_i$ is increased by a factor $O(\log m)$ and
we combine sequentially $2^p=O(m)$ modified automata $C'_i$.
Let prove that the resulting automaton $\tilde C$ gives the required
reduction of the CHP to the $P_{\mathbb B}$-realizability problem.

Suppose  that for some $n$ all constraints in an instance of the CHP
are satisfied. Then there exists a permutation of block words of the
length $\ell{}n$, which satisfies the accepting requirements of the
automaton $C_i$ for each $1\leq i\leq m$. Extending blocks by
prefixes $z_i$ that represent $i-1$ in binary and arranging
permutations in the order of~$i$ we obtain the permutation word
satisfying all three requirements stated above. Hence this word is
accepted by $\tilde C$. 

On the other hand, assume that $\tilde C$ accepts some word $W$ from the
permutation filter. The counter $\tilde C''$ accepts. So
the block rank of $W$ is  $p+\ell{}n$. The automaton $\tilde C'$
accepts also. It follows from the construction that it must pass
through all automata $\tilde C'_i$. There are $2^p$ possible
prefixes of the length $p$. So the construction of modified automata
implies that the sequence of block prefixes of the length $p$ is
uniquely determined. Namely, it has a form
$$
\underbrace{\overbrace{0000\dots0}^{\text{\footnotesize $p$ bits}},
\dots,
\overbrace{0000\dots0}^{\text{\footnotesize $p$
    bits}}}_{\text{\footnotesize $2^{\ell{}n}$ prefixes}},
\dots\kern1em\dots,
\underbrace{\overbrace{111\dots1}^{\text{\footnotesize $p$ bits}},
\dots,
\overbrace{111\dots1}^{\text{\footnotesize $p$
    bits}}}_{\text{\footnotesize $2^{\ell{}n}$ prefixes}}.
$$
Thus
acceptance by $\tilde C'$ means that  suffixes corresponding to each
prefix $z$, where $z$
is a binary representation of~$i-1$, form 
a permutation of the binary words of the length $\ell{}n$ and the
corresponding permutation
word is accepted by the automaton $C_i$. The latter implies that
$i$th constraint in the instance of the CHP is satisfied. It implies that the
answer for the instance of the CHP is positive.

\section{Reduction of the $P_{\mathbb B}$-realizability problem\\ to the CHP}
\label{Filter->Qhit}

The reduction is composed from several intermediate reductions. 

At first, we show that  the $P_{\mathbb B}$-realizability
problem is reducible to the Walk Weight Hitting Problem. The walks in
this problem are walks in a digraph.

Let  $\Gamma(V,E)$ be a digraph. We assume that the edges of the
digraph are colored
in~$s$ colors from the set $\{1,2,\dots,s\}$. 

For a walk $\tau$ we define a \emph{weight}
$w(\tau)\in {\mathbb Z}^s$ as an $s$-dimensional integer vector
\begin{equation*}
  w(\tau)=(c_1(\tau),c_2(\tau),\dots, c_i(\tau),\dots c_s(\tau)),
\end{equation*}
where $c_i(\tau)$ equals the number of edges colored in the
color~$i$ on the walk $\tau$.

\textbf{Walk Weight Hitting Problem} (WWHP)

{\sc INPUT:}
a digraph $\Gamma$, an $s$-coloring of the edges of the $\Gamma$, 
two vertices  $a$, $b$ of   $\Gamma$,
an integer matrix $\Phi$ of order~$s\times s$ and an integer vector $x_0$
of dimension~$s$. 

{\sc OUTPUT:}
`yes' if  the orbit
$\mathop{\mathrm{Orb}}\nolimits_\Phi{}x_0$ 
intersects a set of weights of the walks from the vertex $a$ to the
vertex~$b$ and `no' otherwise.

\medskip

\begin{lemma}\label{PermF->hitRoute}
  The $P_{\mathbb B}$-realizability problem is Turing reducible to the
  \textup{WWHP} problem.
\end{lemma}

The proof of Lemma~\ref{PermF->hitRoute} is presented in
Subsection~\ref{route}. 

Then we will reduce 
the WWHP 
to the  Integer  cone Hitting
Problem. \emph{An integer cone} ${\mathbb N}(v_1,\dots,v_r)$ is the
set of vectors 
$$
\sum_{i=1}^ra_iv_i,\quad  a_i\in\NN,
$$
where $v_i\in\ZZ^d$. We denote the set of nonnegative integers by
$\NN$. 

\textbf{Integer  cone Hitting Problem} (IHP)

{\sc INPUT:}
a square matrix $\Phi$ of order $d$; 
a $d$-dimensional vector $x_0$; 
a family of vectors $v_{i}\in \ZZ^d$, $i=0,1,\dots,r$. 

{\sc OUTPUT:}
`yes' if  the orbit
$\mathop{\mathrm{Orb}}\nolimits_\Phi{}x_0$ 
intersects the translate of the integer cone $v_0+{\mathbb N}(v_1,\dots,v_r)$
and `no' otherwise.

\medskip

\begin{lemma}\label{WWHP->IHP}
  The \textup{WWHP}  is Turing reducible to the \textup{IHP}.
\end{lemma}

The proof of Lemma~\ref{WWHP->IHP} is presented in
Subsection~\ref{integer-cones}. 

Next step is to reduce the IHP 
to the Polyhedral cone Hitting
Problem. This last intermediate problem is stated as follows.

\textbf{Polyhedral cone Hitting Problem} (PHP)

{\sc INPUT:}
a square matrix $\Phi$ of order $d$; 
a $d$-dimensional vector $x_0$; 
a family of linear functions $h_{j}$ 
on ${\mathbb Q}^d$; a shift vector $v_0\in \QQ^d$.

{\sc OUTPUT:}
`yes' if  the orbit
$\mathop{\mathrm{Orb}}\nolimits_\Phi{}x_0$ 
intersects the translate $v_0+K$ of the closed polyhedral cone
 \begin{equation*}
   K=\{x\in {\mathbb Q}^d: h_j(x)\geq 0\}
 \end{equation*}
and `no' otherwise.

\medskip

The PHP problem is Turing-reducible to the ODP problem. The proof is
similar to the proof of Proposition~\ref{CHP=ODP}.

So, the total chain of reductions is
$$
\text{$P_{\mathbb B}$-realizability}\leq_T\text{WWHP}\leq_T
\text{IHP}\leq_T\text{PHP}\leq_T\text{ODP}\leq_T\text{CHP}.
$$

As 
mentioned above the last two reductions are based on
Proposition~\ref{CHP=ODP}. So, to complete the reduction we should
prove the following theorem.

\begin{theorem}\label{Zhit->Qhit}
  The \textup{IHP} problem is Turing reducible to the \textup{PHP}
  problem. 
\end{theorem}

The proof of Theorem~\ref{Zhit->Qhit} is
contained in Subsection~\ref{Zhitsubsect}.

\subsection{Proof of Lemma~\ref{PermF->hitRoute}}\label{route}

We start from some preliminary work.

Let $R$ be a regular language over the alphabet $\{0,1,\#\}$ and let $A$
be a deterministic automaton with the state set $Q$ accepting the
language~$R$. 
Our goal is to check whether $R\cap P_{\mathbb B}\ne\varnothing$.
Let pass to the transition monoid and
express this condition in terms of three maps $f_0$, $f_1$,
$f_\#$ of the set $Q$ into itself induced by
reading  respective symbols.

Define $f(w)$, where $w=w_1w_2\dots w_\ell\in \{0,1\}$, as
$$
\prod_{i=1}^\ell f_{w_{\ell-i+1}}.
$$

We denote the initial state of the automaton by  $q_s$ and the
accepting set by $Q_a$. The reduction algorithm checks for each
accepting state $q_f\in Q_a$   
whether it can be reached from the initial state $q_s$ after
reading some word from the permutation filter. Formally this condition
means that for some word $\#w_1\#{}w_2\#{}\dots
w_N\#\in P_{\mathbb B}$  we have
\begin{equation}\label{hit-acc}
  q_f=\left(\prod_{i=0}^{N-1} f_\# f(w_{N-i})\right)f_\# q_s.
\end{equation}

It is important that to verify the condition~\eqref{hit-acc} we do not
need to know the sequence $w_i$. It is sufficient to compute for each
map $g\in Q^Q$ the number $\nu_n(g)$ 
of its representations in the form $f(w)$, where
$w\in \{0,1\}^n$ (recall that $N=2^n$). Then the
condition~\eqref{hit-acc} is rewritten as
\begin{equation}\label{hit-acc2}
  q_f=\left(\prod_{i=0}^{N-1} f_\# g_i\right)f_\# q_s,  
\end{equation}
where each map $g\in Q^Q$ occurs exactly
$\nu_n(g)$ times in the sequence $g_i$.

It turns out that the integers $\nu_n(g)$ can be expressed as the coordinates
of the vectors taken from the orbit of a linear map.

In $\QQ$-vector space $\QQ(Q^Q)$ equipped with the  basis $\{e(f)\}$ 
indexed by maps $f\colon
Q\to Q$ we define a linear map by the action on the basis vectors
\begin{equation}\label{Phi-def}
  \Phi e( f) = e(ff_0) + e({ff_1}).
\end{equation}
(Recall  that map composition is taken from the right to the
left.)

\begin{prop}\label{2Phi}
  $\nu_n(g)$ equals the $e(g)$-coordinate of the vector
  $\Phi^ne(\mathop{\mathrm{id}})$ w.r.t. the basis  $\{e(f)\}$.
  Here $\mathop{\mathrm{id}}$ is the identity map.
\end{prop}
\begin{proof}
  By induction on $n$. The case of $n=0$ is trivial. 
  If the statement holds for  $n=k-1\geq0$ then it holds for~$n=k$:
  \begin{multline*}
    \Phi^ke(\mathop{\mathrm{id}})=\Phi \sum_{g\in Q^Q}\nu_{k-1}(g) e(g)= 
    \sum_{w\in\{0,1\}^{k-1}}\Phi e(f(w))=\\=
    \sum_{w\in\{0,1\}^{k-1}}(e(f(w)f_0)+e(f(w)f_1)=
    \sum_{w\in\{0,1\}^{k-1}}(e(f(w0))+e(f(w1))=\\=
    \sum_{w\in\{0,1\}^k}e(f(w))=\sum_{g\in Q^Q}\nu_{k}(g)  e(g).
  \end{multline*}
\end{proof}

Now we are going to describe the condition~\eqref{hit-acc2} in terms
of walk weights for  a suitable graph. 
 
For this purpose we will use the Cayley graphs  for monoids.

Let $G=\{g_1,\dots,g_m\}\subseteq Q^Q$ be a set of maps. It generates
a monoid $M$ 
 (a~monoid operation is a map composition, recall that by
definition a~monoid contains also the identity map.). 
By definition the vertices of \emph{the Cayley graph} $\Gamma_G$ of
the monoid are elements of the monoid $M$ and (directed) edges have the
form $(h, g_ih)$ for $h\in M$, $g_i\in G$. Note that the edges of the
Cayley graph are naturally colored by elements of~$G$. 
The Cayley graph of a monoid may contain parallel edges as the equality
$g_ih=g_jh$ for $i\ne j$  is possible for a monoid.

Let $M_{01}$ be a semigroup generated by the maps $f_0$, $f_1$ of the
automaton~$A$. This semigroup is finite and can be
described as the least set of maps from the state set $Q$ to itself
that contains the maps $f_0$, $f_1$ and is closed
w.r.t. multiplication by $f_0$, $f_1$. So, the semigroup $M_{01}$ can
be constructed efficiently\footnote{Hereinafter `efficiency' means a
mere existence of an algorithm.}.

In a similar way we define a monoid $M$ generated by maps $f_\#f$, $f\in
M_{01}$.  
This monoid is also constructed efficiently. 

Denote  by $\Gamma_M$ the Cayley graph of the monoid $M$ w.r.t. the set of
generators $\{f_\#f,\ f\in M_{01}\}$.

It follows from the construction that~\eqref{hit-acc2} holds iff there
exist an integer $n$ and a walk $\tau$ on the digraph $\Gamma_M$ from
the vertex $\mathop{\mathrm{id}}$ to a vertex $h$ such that
$h(f_\#(q_s))=q_f$ and $(w(\tau))_{f_\#g}=\nu_n(g)$ holds for all $g$.
And we conclude that the condition~\eqref{hit-acc2} is equivalent to
the condition
\begin{equation}\label{Cayley-hit}
  w(\tau)=\Phi^n x_0.
\end{equation}

\begin{remark}
  Note that the size of the monoid $M$ can be less than $|Q|^{|Q|}$
  and $\Phi$ acts on $\QQ(Q^Q)$. To
  fix a difference in vector dimensions we extend the coordinates of 
  walk weights by zero values for $e(g)$ such that $g\notin M$.
\end{remark}

Thus the reduction algorithm solves with help of a WWHP-oracle several
instances of the WWHP indexed by the accepting states $q_f$ and 
mappings
$h$
such that\linebreak $h(f_\#(q_s))=q_f$.

An instance has the following input data: a digraph is the Cayley
graph $\Gamma_M$, colors are generators, the initial vertex is the
identity map 
$\mathop{\mathrm{id}}$, the terminal vertex is $h$, the matrix is
defined by~\eqref{Phi-def} and the initial vector is
$x_0=e(\mathop{\mathrm{id}})$.

If the answer is positive for some instance from the described set 
then the answer in the  instance of the $P_{\mathbb
B}$-realizability problem involved is also positive due to
Proposition~\ref{2Phi} and condition~\eqref{Cayley-hit}.

Otherwise,
it follows from the definition
of the graph $\Gamma_M$ and Proposition~\ref{2Phi} that
the answer in the instance of the $P_{\mathbb
B}$-realizability problem is negative. 

Lemma~\ref{PermF->hitRoute} is proved.

\subsection{Proof of Lemma~\ref{WWHP->IHP}}\label{integer-cones}

We need the following 
%5-12
auxillary
lemma.

\begin{lemma}\label{basic}
  Let $a,b$ be vertices of a digraph $\Gamma$ colored in $s$ colors
  and let $W(a,b)$ be the set of weights of all walks from the vertex~$a$
  to the vertex~$b$.  The set $W(a,b)$ is a finite union of translates
  of integer cones  in  ${\mathbb Z}^s$. 

  The list of the cones and the vectors of translation is constructed
  efficiently. 
\end{lemma}
\begin{proof}
  Writing  a sequence of colors along a walk gives a word in the $s$-letter
  alphabet. 
  The collection of such words for all walks from the vertex
  $a$ to the vertex~$b$ forms a regular language.  

  The weight of a walk is just the Parikh image of the
  corresponding word. So the statement of the lemma follows from
  Parikh's theorem~\cite{Parikh}.
\end{proof}

Lemma~\ref{basic} implies that to check the
condition~\eqref{Cayley-hit} it is sufficient to check several conditions
of the form
\begin{equation}\label{Zhit-def}
  \Phi^nx_0\in v_0+W,
\end{equation}
where $v_0$ in an integer  $s$-dimensional vector and $W$ is an
integer cone in~${\mathbb Z}^s$. It gives a reduction of the WWHP to
the IHP.

\subsection{Proof of Theorem~\ref{Zhit->Qhit}}\label{Zhitsubsect}

Let $W_{\mathbb Q}\subseteq {\mathbb Q}^s$ be a rational cone
generated by vectors $v_1$, $\dots$, $v_r$. There is an algorithm
(see, e.g.~\cite[\S7.2]{Schrijver}) for passing to dual polyhedral
description of $W_{\mathbb Q}$.

The cone $W_{\mathbb Q}$ may contain some additional integral points
not belonging to 
integral cone $W$ so that conditions~\eqref{Zhit-def} should be modified.

Let us start with a simple case of an integral simplicial cone, when
generating vectors $v_i$ are linear independent.

\begin{prop}\label{simplex}
  If vectors $v_1$, \dots, $v_r$ 
  are linear independent then representation
  $x=\sum_{i}a_iv_i$, $a_i\in{\mathbb N}$ holds if and only if
  the following two representations hold:
  $x=\sum_{i}b_iv_i$, $b_i\in{\mathbb Q}_+$ and
  $x=\sum_ic_iv_i$, $c_i\in{\mathbb Z}$.  
\end{prop}
\begin{proof}
The implication $\Leftarrow$ is trivial and the opposite implication 
$\Rightarrow$ follows from linear independence of vectors $v_i$,
as the equality 
  $$
  \sum_{i}b_iv_i=\sum_ic_iv_i
  $$
holds iff $b_i=c_i, \, i=1,\dots, r$. 
\end{proof}

Vector $x$ can be represented in a form
$x=\sum_ic_iv_i$, where $c_i\in{\mathbb Z}$, iff $x$ is an element of the
subgroup $G$ generated by vectors~$v_i$ in ${\mathbb Z}^s$.  And this
condition in turn is a conjunction of two conditions. The first
condition is simple: $x$ is in a subspace generated by vectors~$v_i$
in the vector space space ${\mathbb Q}^s$.  To formulate the second
condition let extend the set $v_i$ of generators of $G$ to the basis
of ${\mathbb Z}^s$ by adding some unit coordinate vectors $e_j$ so
that the resulting system formed from $v_i$ and $e_j$ constitute a
basis of ${\mathbb Z}^s$.  Such extension is always possible in view
of linear independence of $v_i$.  Let $\tilde G$ be the group generated
by $v_i$ and $e_j$.

\begin{prop}\label{subspace}
  Vector $x$ belongs to the subgroup $G$ generated by vectors~$v_i$ in
  the group ${\mathbb Z}^s$ iff $x$ is in the subspace generated by
  vectors~$v_i$ in ${\mathbb Q}^s$ and $x$ belongs to the subgroup
  $\tilde G$.
\end{prop}

\begin{proof}
  The implication $\Leftarrow$ is obvious. To prove the opposite implication 
  $\Rightarrow$  assume that
  $x\in \tilde
  G$. By construction of the subgroup~$\tilde G$ it follows that
  $$
  x=g+\sum_j x_je_j,
  $$ where $g\in G$ and $e_j$~denote those added coordinate unit
  vectors.  Now it follows from linear independence of $v_i$ and $e_j$
  that if vector $x$ belongs to the subspace generated by vectors~$v_i$
  in the coordinate vector space ${\mathbb Q}^s$ then all coordinates
  $x_j$ are zero. Hence $x$ belongs to the subgroup~$G$.
\end{proof}

It follows from the propositions~\ref{simplex} and~\ref{subspace} that
a vector falls into integral conic hull of vectors $v_i$ if and only
if three conditions are fulfilled. The first two conditions are
${\mathbb Q}$-linear and control whether a vector falls into the cone
$W_{\mathbb Q}$.  The third condition consists in checking that a
vector belongs to a
fulldimensional lattice in ${\mathbb Z}^s$ (a subgroup of ${\mathbb Z}^s$
having finite index in 
${\mathbb Z}^s$).  Checking the first and the second condition for the
orbit of a linear map are particular variants of CHP and the third
condition holds for a nice family of  orbit points.

\begin{prop}\label{group}
  Let $G=\langle v_1,\dots, v_s\rangle\subseteq{\mathbb Z}^s$ be
  rank~$s$ subgroup of ${\mathbb Z}^s$, let $v_0$ be an integral
  vector in ${\mathbb Z}^s$ and let $\Phi$ be an integral $s\times s$
  matrix.

  Then the set
  $$
  H=\{n: \Phi^nx_0\in v_0+G\}
  $$ 
  is a union of a finite set $H_0$ and a finite set of arithmetic
  progressions. Moreover, there is an algorithm to compute $H$.
\end{prop}
\begin{proof}
First let $G$ be a diagonal subgroup whose generators are columns of
the matrix  
  $$
  D=\begin{pmatrix}
    q_1&0&\dots&0\\
    0  &q_2&\dots&0\\
    \hdotsfor{4}\\
    0 &0&\dots&q_s
  \end{pmatrix}
  $$
The shift induced by $G$ is given by
  \begin{equation}\label{cond}
    x_i\equiv   \xi_i\pmod{q_i},\quad 1\leq i\leq s.
  \end{equation}

Let $q$ be the least common multiple of $q_i$.  Now calculating
$\Phi^kx_0$ modulo $q$ we can find a finite set of exponents $k$, and
the corresponding arithmetic progressions $k\pmod{q}$,
satisfying~\eqref{cond}.  Besides these arithmetic projections $H$ may
include only numbers that should be less than initial terms of the
progressions found on the previous step, so that we can directly
compute this finite set.

The general case is reduced to a diagonal one by reducing the generator matrix
$M$ to the Smith normal form  
(see~\cite[\S4.4]{Schrijver}) If $G$ is a subgroup of ${\mathbb Z}^s$
of full rank then Smith normal form is given by
  \begin{equation}\label{Smith}
    M=UDV, 
  \end{equation}
where $U$, $V$~are unimodular matrices and $D$ is a nondegenerate
diagonal matrix.  $V$ corresponds to column transformations and gives
a system of new generators $\tilde v_1$, \dots, $\tilde v_s$ of~$G$.
$U$ corresponds to row transformations and shows how the vectors $e_i$
of the initial basis can be expressed as a linear combinations of the
new basis vectors $\tilde e_i$:
  $$
  (e_1,\dots, e_s)=(\tilde e_1,\dots \tilde e_s) U.
  $$
  By~\eqref{Smith} 
the generator matrix 
$\tilde v_i$ of the group $G$ is diagonal in the basis
$\tilde e_i$.

To complete a reduction to the diagonal case we should express the matrix
$\Phi$ and the initial vectors $x_0$, $v_0$ in the new basis $\tilde
e_i$.
\end{proof}

\begin{lemma}\label{Zsimplexhit->Qhit}
  If vectors  $v_1$, \dots,
  $v_r$ are linear independent then the \textup{IHP} 
  is reduced to the \textup{PHP}.
\end{lemma}
\begin{proof}
  Let us describe an algorithm with PHP-oracle that solves the IHP for
  instances of 
  linear independent set of vectors
  $v_1$, \dots,  $v_r$.

  First let find the set $H$ defined in proposition~\ref{group}. We
  should check whether condition~\eqref{Zhit-def} holds for $H$. The
  finite part $H_0\subseteq H$ can be checked directly.  And for each
  arithmetic progression $n=n_0+Nk$, $k=0,1,\dots$ in $H$ we compute a
  vector
  $x_1=\Phi^{n_0}$  and a matrix
  $\Phi_1=\Phi^N$. 
  Then using PHP-oracle we check whether there exists $k$ such that
  $\Phi_1^kx_1-v_0$ 
  hits the rational cone~$W_{\mathbb Q}$. 
  If all such tests fail we answer `no'. Otherwise we answer `yes'.

To check correctness of the algorithm involved note that it follows
from the propositions~\ref{subspace} and \ref{group} that all tests
fail if and only if the orbit $\Phi^nx_0$ either does not intersect
the cone~$W_{\mathbb Q}$ at all or hits it for exponents~$n$ such that
$\Phi^nx_0-v_0$ is not in the group~$G$.  Hence, the answer in the IHP
is negative.

On the other hand, if some test is successful then 
we find~$n$, such that   $\Phi^nx_0-v_0$ is in the group $G$ and hits the 
cone~$W_{\mathbb Q}$. 
It follows from proposition~\ref{simplex} 
that the answer in the  IHP is positive.
\end{proof}

Now we discuss the general case.  Let 
$$W_{\mathbb N}={\mathbb N}(v_1,\dots,v_r)$$ 
be the set of integral linear combinations
of~$v_i$ with nonnegative coefficients (the integral cone generated
by~$v_i$). Below we describe an algorithm that given any integral cone
$W_{\mathbb N}$ produces its representation as a union of a finite set
$W_0$ of singular points and a finite set of translates of integral
simplicial cones $u_i+{\mathbb N}(v_{i_1},\dots,v_{i_t})$.

Recall that Caratheodory's theorem (see, e.g.~\cite[\S7.7]{Schrijver})
states that any vector in ${\mathbb Q}_+(v_1,\dots, v_r)$ 
is a rational nonnegative linear 
combination of some $\leq s$ (where $s$ is
the dimension) vectors from the system
$v_1$, \dots, $v_r$.

As one can form not more than $\binom rs$ systems of linear
independent vectors out of $r$ vectors, then it is sufficient to
consider the case of the intersection of the integral cone $W_{\mathbb
N}$ and a rational simplicial cone with $s$ generators chosen from the
set $v_1,\dots, v_r$.

Let $K$ be one of such simplicial cones and let $\tilde
v_1,\dots,\tilde v_t$ be the set of its generating vectors.  $K$
contains ${\mathbb N}(\tilde v_1,\dots,\tilde v_t)$, but besides that
$K$ may contain some extra integral points.  Now we use the fact that
$K$ has \emph{Hilbert basis\/}: a set of integral vectors $u_1$,
\dots, $u_m$, such that
\begin{equation}\label{Hilbert}
  K\cap{\mathbb Z}^s={\mathbb N}(u_1,\dots,u_m).
\end{equation}
In particular, it follows from~\eqref{Hilbert} 
that
$$
K\cap{\mathbb Z}^s=\bigcup_{i=1}^m\big(u_i+{\mathbb N}(\tilde
v_1,\dots,\tilde v_t)\big). 
$$

Recall that one can effectively find Hilbert basis for~$K$
as it coincides with the set of integral points of the polytope
\begin{equation*}
  \{\lambda_1\tilde v_1+\lambda_2\tilde v_2+\dots+\lambda_t\tilde v_t : 0\leq
  \lambda_i\leq 1\}.
\end{equation*}

Some vectors from Hilbert basis belongs to the subgroup ${\mathbb
Z}^s$, generated by $v_i$, i.e. to the integral hull of $v_i$.
Checking this condition is effective as it is reduced to solving linear
Diophantine equations.

Let take one of such vectors
$u=\sum_{i}b_iv_i$ and let
$B$ be maximum of modules of the coefficients
$b_i$. 
Now all points of the intersection of 
$u+{\mathbb N}(\tilde v_1,\dots,\tilde v_t)$ and $W_{\mathbb N}$ belong to the 
union of the finite set
\begin{equation}\label{first}
\{x: x=u+\sum_{i=1}^ta_i\tilde v_i,\ 0\leq a_i\leq B\},
\end{equation}
the integral cone
\begin{equation}\label{second}
  u+(B+1)\sum_i\tilde v_i+{\mathbb N}(\tilde v_1,\dots,\tilde v_t)
\end{equation}
and  $(B+1)t$ sets of the form
\begin{equation}\label{rest}
  \big(u+a\tilde v_i+{\mathbb Q}_+(\tilde v_1,\tilde v_2,\dots, 
  \tilde v_{i-1}, \tilde v_{i+1},\dots, \tilde v_t)\big)\cap
  W_{\mathbb N},\quad
  0\leq a\leq B,\ 1\leq i\leq t.
\end{equation}
Each set in~\eqref{rest} is an intersection of the integral cone
$W_{\mathbb N}$
and a translate of some rational simplicial cone of the lower
dimension
$t-1$.

Now we can give a complete description of the algorithm that given an
integral cone $W_{\mathbb N}$ and a translate of some rational
simplicial cone whose generating vectors are taken from the set of
generators of $W_{\mathbb N}$, finds representation of the
intersection of the cones involved as a union of a finite set of
singular points~$W_0$ and a finite set of translates of simplicial
integral cones.

The procedure is recursive.  Applying it to the intersection of
$W_{\mathbb N}$ and a $t$-dimen\-sional rational simplicial cone $K$ we
at first compute Hilbert basis of~$K$. Then for any vector in Hilbert
basis that belongs to the integral hull of vectors $v_1$, \dots,
$v_r$, we find the sets~\eqref{first},~\eqref{second}
and~\eqref{rest}.

The set~\eqref{first} is added to the singular part~$W_0$.  The
simplicial integral cone~\eqref{second} is included to the family of
integral cones computed on the previous steps and the algorithm is
continued recursively by processing all sets~\eqref{rest}.

The procedure is finite as the sets~\eqref{rest} are empty for any
one-dimensional cone.  
Hence, the recursion tree has finite height and finite degree at any point.

Thus the following theorem is proved.

\begin{theorem}\label{FrobNdim}
An integral translate of an integral cone
$$v_0+{\mathbb N}(v_1,\dots,v_r)$$ 
can be represented as a union of a finite set and a finite family of translates
of simplicial integral cones. There is an algorithm that finds such a representation
from the list of vectors 
$v_0$, $v_1,\dots,v_r$.
\end{theorem}

Now the main result of this section easily follows.

\begin{proof}[Proof of theorem~\textup{\ref{Zhit->Qhit}}]
  Using algorithm from Theorem~\ref{FrobNdim} find representation 
  of the integer cone  
  as a union of a finite set of singular points
  and finite family of translates of simplicial integral cones. Now
  for any cone in the family check whether the orbit hits it using
  PHP-oracle from Lemma~\ref{Zsimplexhit->Qhit}.

  Checking whether or not the orbit hits a point (and hence, 
  any finite set of points) is
  also reduced to the PHP as a point may be regarded as a translate of
  the zero cone.
\end{proof}

\section{Decidable and undecidable variants of the regular
  realizability problem}\label{decid-and-undec}

The Skolem problem is open for almost eighty years.  Using slight
abuse of language, presently it falls `on the border between
decidability and undecidability' \cite{Halava05}. In 
this
section
we show that analogous `borderline' pattern holds for a more general
\fP-realizability problem and give some 
decidable and undecidable problems closely related to it.
All these 
problems
are problems of regular realizability. The languages 
specifying the problems consist of binary \emph{block words} separated by the
delimiter $\#$. All blocks have the same length (the block rank $n$). 
Blocks of a block word form a multiset of binary words of the
same length. We will define languages indicating the properties of
this \emph{block multiset\/}.

Now we present decidable examples. 

\emph{The surjective filter} \fS{} consists of those block words  which
  block multiset contains all words of the length $n$.

\emph{The injective filter} $\fI$ consists of those block words  which
  block multiset is a set, i.e. each block appears at most once in a
  word from \fI.

It is clear from the definitions that  $\fI\cap\fS=\fP$.

The problems of \fI-realizability and \fS-realizability are
decidable. Let's outline the proof.  It turns out that 
both
\fI-
or, respectively, \fS- realizability 
can be reduced
to some restricted versions of the
integer cone hitting problem
whose decidability 
follows from specific properties of maps~$\Phi$ defined in
Subsection~\ref{route}.

We use the componentwise partial order on the integer orthant $\ZZ_+^d$
\begin{equation}
  x\prec y\ \Leftrightarrow \ x_i\leq y_i\ \text{for all }i.
\end{equation}

\medskip

\textbf{Up-hitting Problem}

{\sc INPUT:}
a square matrix $\Phi$ of order $d$; 
a $d$-dimensional vector $x_0$; 
a family of vectors $v_{i}\in \ZZ^d$, $i=0,1,\dots,r$. 

{\sc OUTPUT:}
`yes' if  the \emph{orbit up-shadow\/} 
intersects the translate of the integer cone $v_0+{\mathbb
  N}(v_1,\dots,v_r)=v_0+W$ 
and `no' otherwise. It means that 
\begin{equation}\label{suporbit}
  \Phi^nx_0\prec y
\end{equation}
holds for some integer~$n$ and $y\in v_0+W$.

The \textbf{Down-hitting problem} is defined similarly except for the
condition~\eqref{suporbit} which is replaced by the condition
\begin{equation}\label{suborbit}
  y\prec \Phi^nx_0.
\end{equation}

\medskip

\begin{lemma}\label{Inj->suborbit}
  The \reg(\fI) problem is Turing reducible to the down-hitting problem. 
\end{lemma}
\begin{proof}
  Repeat the arguments from the proof of
  Lemma~\ref{PermF->hitRoute}. Now it is possible that 
  some binary words  in a block word are missed. It means that the
  condition~\eqref{Cayley-hit} is replaced by 
  \begin{equation}\label{down-hit}
    w(\tau)\prec\Phi^n x_0.
  \end{equation}
  Applying the arguments from Lemma~\ref{WWHP->IHP} we see
  that~\eqref{down-hit} is transformed to~\eqref{suborbit} for cones
  appeared in the reduction from  Lemma~\ref{WWHP->IHP}. 
\end{proof}

In a similar way we reduce the surjective filter.

\begin{lemma}\label{Sur->suporbit}
  The \reg(\fS) problem is Turing reducible to the up-hitting problem. 
\end{lemma}
\begin{proof}
  In a block word taken from the \fS{} all binary words of the
  length~$n$ appear (possibly, several times). It leads to the
  condition~\eqref{suporbit} for cones appeared in the reduction from
  Lemma~\ref{WWHP->IHP}.  
\end{proof}

Dickson's lemma 
claims that there are no infinite antichains in the poset
$(\ZZ^d_+,\prec)$. So an orbit up-shadow is a finite union of
translated copies of the orthant $\ZZ^d_+$. In the case of an orbit
down-shadow copies of the orthant are replaced by  `parallelepipedons'
(the Cartesian products of segments). Thus the up-hitting problem as
well as the down-hitting problem is reduced to the nonemptiness check
for intersections of integer cones (parallelepipedons), which is an
integer linear
programming problem, provided the representations
mentioned above can be constructed efficiently.  

To construct the aforementioned representations we use specific
properties of asymptotic behavior of the orbit points~$\Phi^nx_0$,
where $\Phi$ is determined by~\eqref{Phi-def}.

Recall that an $e(g)$-component of $\Phi^nx_0$ is expressed as the
number $\nu_n(g)$ of walks of the length~$n$ from the vertex $\id$ to
the vertex $g$ in the graph $\Gamma$.  The vertex set of the graph
$\Gamma$ is $V(\Gamma)=Q^Q$ and the edge set $E(\Gamma)$ consists of
pairs in the form $(f,ff_0)$ or $(f,ff_1)$.  (See
Subsection~\ref{route} for a detailed exposition.)

Note that the integer $\nu_n(g)$ is the number of words of the length $n$
in a regular language. This language is recognized by an automaton with
the transition graph $\Gamma$. Thus the generating function 
$$
\ph_g(t)=\sum_{n=0}^\infty \nu_n(g)t^n
$$
for the
language is a rational function and its representation in the form
$P(t)/Q(t)$ can be found efficiently.  

In the arguments below we need some properties specific to generating
functions of regular languages.
So it is more suitable for our purposes to analyze the asymptotic
behavior in combinatorial settings by considering walks on the
graph $\Gamma$.

\begin{prop}\label{1-letterPos}
  For any vertex $g\in V(\Gamma)$ the set
  $$
  P_g=\{n: \nu_n(g)>0\}
  $$
  is a semilinear set (a finite union of an exceptional finite set and
  finite collection of arithmetic progressions) and its description
  can be constructed efficiently.
\end{prop}
\begin{proof}
  Regard the $\Gamma$ as the transition graph of a nondeterministic
  automaton over a 1-letter alphabet. Then the proposition follows from
  Parikh's theorem.
\end{proof}

\begin{remark}\label{divisors}
  Differences of all progressions in Proposition~\ref{1-letterPos} are
  the cycle lengths in the graph $\Gamma$. So they are divisors of the 
  least common multiple of integers from $1$
  to $ |V(\Gamma)|$. In 
the sequel
use we denote this common multiple
  by~$N$.
\end{remark}

Take a vertex $g$ on a directed cycle of the length $\ell$. The
following inequality
\begin{equation}\label{monotone}
  \nu_{n+\ell}(g)\geq\nu_n(g).
\end{equation}
holds. Indeed, one can extend any walk of the length $n$ by the cycle. 

Now we divide the vertices of the graph $\Gamma$ into three groups.
\begin{itemize}
\item [--] $V_1$ consists of vertices $v$ such that some directed cycle
  (possibly, a loop) goes through the vertex $v$.
\item [--] $V_2$ consists of vertices $v$ such that there is a walk
  starting at the $\id$, finishing at the $v$ and passing through a
  vertex from the set~$V_1$.
\item [--] $V_3$ consists of all other vertices.
\end{itemize}

\begin{prop}\label{zero-set}
  If $g\in V_3$ then $\nu_n(g)=0$ for $n>|V(\Gamma)|$.
\end{prop}
\begin{proof}
  Any walk of the length $n>|V(\Gamma)|$ from $\id$ to $g$ must
  contain repeating vertices. It means that a part of the walk is a
  directed cycle. So the walk passes a vertex from the set $V_1$.
\end{proof}

\begin{prop}\label{uniform-monotone}
  The inequality
  \begin{equation}\label{N-monotone}
  \nu_{n+N}(g)\geq \nu_n(g),
  \end{equation}
  where $N=\mathrm{LCM}(1,\dots,|V(\Gamma)|$,
  holds for all
  $g\in V(\Gamma)$ and $n>|V(\Gamma)|$.
\end{prop}
\begin{proof}
  For $g\in V_3$ apply Proposition~\ref{zero-set}.

  For $g\in V_1$ the inequality~\eqref{N-monotone} follows
  from~\eqref{monotone} and Remark~\ref{divisors}.

  Now take a vertex $g\in V_2$. The set $V_g\subseteq V_1$ consists of
  vertices $g'\in V_1$ such that $g'$ belongs to a walk from~$\id$ to
  $g$ and for all walks of this type all vertices after the $g'$ along a walk
  are in the set~$V_2$. The subgraph $\Gamma_g$ is induced by the
  edges of all walks from a vertex in $V_g$ to $g$. 

  Let observe the following properties of the $\Gamma_g$.

  There are no edges to vertices of the set $V_g$ in the graph
  $\Gamma_g$. Indeed, such an edge contradicts definition of the set~$V_g$.

  There are no edges outgoing from the vertex $g$ in the graph
  $\Gamma_g$. Otherwise one would detect a directed cycle passing
  through $g$.

  From these properties we conclude that the graph  $\Gamma_g$ is
  acyclic. By definition there are no directed cycles passing through
  vertices in the set $V_2$. Other vertices in the $\Gamma_g$ are in
  the set  $V_g$. There are no directed cycles passing through these
  vertices because there are no edges ingoing  to them.

  Note also that the maximum of the path length from $g'\in V_g$ to
  $g$ does not exceed  $|V(\Gamma)|$.

  From all the properties above we get
  \begin{equation}\label{V2-rec}
    \nu_n(g)=\sum_{\substack{g'\in V_g\\ k\leq |V(\Gamma)|}}
      p_{g',k}\nu_{n-k}(g'),
  \end{equation}
  where $p_{g',k}$ is the number of paths from $g' $ and $g$ in
  the~$\Gamma_g$.  

  Applying the inequality~\eqref{N-monotone} to all terms in the
  right-hand side 
  of~\eqref{V2-rec} we get the same inequality for the left-hand side,
  i.e. for the vertex~$g$.
\end{proof}

\begin{theorem}
  The \reg(\fS) problem is decidable.
\end{theorem}
\begin{proof}
  It follows from Lemma~\ref{Sur->suporbit} that it is enough to
  construct an integer cone representation for an orbit up-shadow.
  Proposition~\ref{uniform-monotone} implies that the orbit up-shadow is
  $$
  \bigcup_{i=0}^N \big(\Phi^ix_0 + \NN^m\big),
  $$
  where   $m$ is a dimension, i.e. the cardinality of $Q^Q$. 

  So the problem is reduced  to the integer linear programming problem.
\end{proof}

To prove decidability of the injective filter we should determine for
each $0\leq r<N$ unbounded components of $\Phi^{nN+r}x_0$ and the
limit values of the bounded components. 

All subsequences $t^r_g(n)=\nu_{nN+r}(g)$, where $0\leq r <N$, are
nondecreasing due to Proposition~\ref{uniform-monotone}. 

A subsequence $t^r_g(n)$, where $g\in V_3$ stabilizes for
$n>|V(\Gamma)|$ and the limit value for it is $0$.

It follows from~\eqref{V2-rec} that a subsequence $t^r_g(n)$, where
$g\in V_2$, tends to
infinity iff at least one of the subsequences $t^{r-k}_{g'}(n)$ tends
to infinity, where $p_{g',k}\ne0$.

The remaining case  $t^r_g(n)$, where $g\in V_1$, is covered by the
following proposition. 

\begin{prop}\label{2infinity}
  Let  $g\in V_1$. Then $\lim\limits_{n\to\infty}t^r_g(n)=\infty$ iff
   there exist a directed cycle~$C$ passing through $g$ and an edge
   $(g',g'')$ such that
   \begin{itemize}
   \item [\textup{(i)}] the edge $(g',g'')$ is not included in the cycle~$C$;
   \item [\textup{(ii)}] the cycle $C$ passes through the vertex $g''$;
   \item [\textup{(iii)}]  $\nu_{nN+r-\ell-1}(g')>0$, where
   $\ell$ is the distance from $g''$ to $g$ along the cycle~$C$.
   \end{itemize}
\end{prop}

Note that due to Remark~\ref{divisors} the conditions (iii) are
equivalent for all~$n$. It is clear that
(i)--(iii) are verified efficiently.

\begin{proof}[Proof of Proposition~\textup{\ref{2infinity}}.]
  `If' part of the proposition follows from
  \begin{equation}\label{increase}
  \nu_{(n+1)N+r}(g)\geq \nu_{nN+r}(g) +\nu_{(n+1)N+r-\ell-1}(g')> 
  \nu_{nN+r}(g).
  \end{equation}

  Prove now `only if'. Suppose that
  $\nu_{nN+r}(g)=T$ for $n>n_0$. For any directed cycle $C$ passing
  through  $g$ and any edge satisfying (i)--(ii) the first
  inequality in~\eqref{increase} implies that
  $\nu_{nN+r-\ell-1}(g')=0$ for $n>n_0$.
\end{proof}

\begin{theorem}
   The \reg(\fI) problem is decidable.
\end{theorem}
\begin{proof}
  Determine all
  unbounded components for all subsequences
  $\Phi^{nN+r}x_0$ and the limit values for bounded components. 
  For this purpose use  Proposition~\ref{2infinity} and
  the observations preceding it.
  Then the down-shadow is the union of the sets
  \begin{equation*}
    \left\{
    \begin{aligned}
      &y_g\geq0,&&\text{if } \lim_{n\to\infty}(\Phi^{nN+r}x_0)_g=\infty ,\\
      y^\infty_g\geq &y_g\geq0, &&\text{if }       
      \lim_{n\to\infty}(\Phi^{nN+r}x_0)_g=y^\infty_g.
    \end{aligned}\right.
  \end{equation*}
  over all   $0\leq r<N$. Here $y_g$ are coordinates in the space $\QQ^{Q^Q}$.

  So the problem is reduced  to the integer linear programming problem.
\end{proof}

Now we present an undecidable realizability problem 
that is 
related to the \fP-realizability problem.

In the construction we use a track product of languages consisting of
block words (\emph{block languages}). Here blocks are words over a
finite alphabet $\Sigma$ and a block word consists of blocks of the
same length separated by the delimiter~$\#$.

Let $\Sigma_1$, $\Sigma_2$ be finite alphabets. For the alphabet
$\Sigma_1\times\Sigma_2$ there are two natural projections from
$(\Sigma_1\times\Sigma_2)^*$ to $\Sigma_1^*$ (resp. to $\Sigma_2^*$):
\begin{equation}
  \begin{aligned}
    &\pi_1\colon (a_1,b_1)(a_2,b_2)\dots(a_n,b_n)\mapsto a_1a_2\dots
    a_n,\\
    &\pi_2\colon (a_1,b_1)(a_2,b_2)\dots(a_n,b_n)\mapsto b_1b_2\dots
    b_n.
  \end{aligned}
\end{equation}

\emph{The track product} of two block languages $L_1$ (over an
alphabet $\Sigma_1$) and $L_2$ (over the alphabet $\Sigma_2$) is a
block language $L=L_1{\parallel}L_2$ consisting of all block words
over the alphabet $\{\#\}\cup\Sigma_1\times \Sigma_2$ such that
the projection $\pi_1$ is in the language $L_1$ and the projection
$\pi_2$ is in the language $L_2$. (For consistency we assume that
$\pi_1(\#)=\pi_2(\#)=\#$.)  

Denote by $\Per_\Sigma$ the block language consisting of all periodic words
over the alphabet $\Sigma$ (all blocks of a word in $\Per_\Sigma$ are equal) 
and by $P_\Sigma$ the block language
consisting of permutation block words over the alphabet 
$\Sigma$ (blocks of a word in $P_\Sigma$ form the set of all words in
$\Sigma^n$, where $n$ is the block rank).

The \reg(\Per_\Sigma) problem is decidable. Actually it is
PSPACE-complete~\cite{Vya09}. It turns out that the track product of the
periodic filter  with the permutation one is undecidable.

\begin{theorem}\label{PePe:undecidable}
  There are alphabets $\Sigma_1$, $\Sigma_2$ such that
  the  \reg(\PePe) problem is undecidable.
\end{theorem}

A suitable undecidable problem that is reduced 
to the \reg(\PePe) problem
 is the following.

\medskip

\textbf{Zero in the Upper Right Corner Problem.} (The ZURC problem.)  
For a given collection of $D\times D$ integer
matrices 
$A_1,\dots, A_N$ check whether the multiplicative semigroup generated
by $\{A_i\}$ contains a matrix $M$ such that $M_{1D}=0$.

In other words the ZURC problem is to check an existence of an integer
sequence  $j_1,\dots,j_\ell$, where $1\leq j_t\leq N$, such that
\begin{equation}\label{zero-expression}
  (A_{j_1}A_{j_2}\dots A_{j_\ell})_{1D}=0.
\end{equation}

\begin{theorem}[see \cite{BellPotapov}]\label{mortality}
  The ZURC problem is undecidable for  $N=2$ and $D=18$.
\end{theorem}

We will reduce the ZURC problem with  $N=2$ and $D=18$ to the
\reg(\PePe) problem where
\begin{equation}\label{alphabets}
  \Sigma_1=[1,\dots,N],\qquad \Sigma_2=[1,\dots,D]\times[1,\dots,D]
  \times\{0,1\}.
\end{equation}

The reduction is similar to the reduction in Section~\ref{HCP}.

Let  $A_1,\dots, A_N$ be an instance of the ZURC problem for
 $D=18$, $N=2$. Rewrite the matrices in the form
$$
A_j=
\begin{pmatrix}
  \eps^j_{11}m^j_{11}& 
  \eps^j_{12}m^j_{12}&
  \dots          &
  \eps^j_{1D}m^j_{1D}\\ 
  \eps^j_{21}m^j_{21}& 
  \eps^j_{22}m^j_{22}&
  \dots          & 
  \eps^j_{2D}m^j_{2D}\\
  \hdotsfor{4}\\
  \eps^j_{D1}m^j_{D1}&
  \dots          &
  \eps^j_{D\;(D-1)}m^j_{D\;(D-1)}& 
  \eps^j_{DD}m^j_{DD} 
\end{pmatrix}\;,
$$
where $m^j_{ik}>0$ and
$\eps^j_{ik}\in\{\pm1,0\}$. Let $M$ be the maximum of  $m^j_{ik}$.

Fix now a sequence $A_{j_1}$, $A_{j_2}$, $\dots$,
$A_{j_\ell}$. 
Matrix elements in the product have the form

\begin{equation}\label{prod-expression}
  \big(A_{j_1}A_{j_2}\dots A_{j_\ell}\big)_{ik}=
  \sum_\tau \eps(\tau)m(\tau),
\end{equation}
where $\tau $ runs over all sequences of  pairs $(i_\al{}k_\al)$ such
that  the length of a sequence is~$\ell$ and 
$i_1=i$, $k_\ell=k$,
$i_{\al+1}=k_\al$,
\begin{equation}\label{prod-element}
\eps(\tau)=\prod_{\al=1}^\ell \eps^{j_\al}_{i_\al{}k_\al},\qquad
m(\tau)   =\prod_{\al=1}^\ell m^{j_\al}_{i_\al{}k_\al}.
\end{equation}

Using the expansion~(\ref{prod-expression}, \ref{prod-element}) we
define the partition of words of the length $\ell\cdot\lceil\log_2
M\rceil$ over the alphabet $\Sigma_1\times \Sigma_2$ into three sets 
$T^+_{ik}(\boldsymbol j)$, $T^-_{ik}(\boldsymbol j)$
and $\Tb_{ik}(\boldsymbol j)$, where   $\boldsymbol j=j_1,\dots,
j_\ell$. (Below we drop out $\boldsymbol j$ while the sequence
$\boldsymbol j$ is fixed.)

It is convenient to represent
a word over the alphabet $\Sigma_1\times\Sigma_2$, where $\Sigma_i$
are given by~\eqref{alphabets}, by a $4$-row table.
A table column represents a symbol in the word. The first row bears
symbols from $\Sigma_1$ while the remaining three columns
represent symbols from
$\Sigma_2$ (they have three components as indicated
in~\eqref{alphabets}). 

A word in the set $T^+_{ik}$ ($T^-_{ik}$) can be divided in  $\ell$
subwords of the length $\lceil\log_2 M\rceil$. The $\al$th subword has
the form
\begin{equation}\label{T-group}
\begin{pmatrix}
  j_\al&j_\al&\dots&j_\al \\
  i_\al&i_\al&\dots&i_\al \\
  k_\al&k_\al&\dots& k_\al \\
  \beta_0&\beta_1&\dots&\beta_{\lceil\log_2 M\rceil-1}
\end{pmatrix}\;.
\end{equation}
in the table representation defined above. The upper three elements in
each row are the same in~\eqref{T-group}. Note that $j_\al$ is the
$\al$th element of the  sequence $\boldsymbol j$. The fourth row
is a binary representation of an integer~$\beta$.

A word in the set $T^+_{ik}$ should satisfy the following requirements
\begin{itemize}
\item [(a)]  $i_1=i$;
\item [(b)] $k_\ell=k$;
\item [(c)] $i_{\al+1}=k_\al$;
\item [(d)] $\beta<m^{j_\al}_{i_\al{}k_\al}$;
\item [(e)] $\eps(\tau)=1$, where $\tau$ is the sequence of pairs
  $(i_\al,k_{\al})$ and $\eps(\tau)$ is given by~\eqref{prod-element}. 
\end{itemize}

A word in the set $T^-_{ik}$ 
should satisfy the requirements (a)--(d) and the modified requirement
(e$'$)
 $\eps(\tau)=-1$.

The set  $\Tb_{ik}$ collects the rest of words.

\begin{prop}\label{pos-expansion}
  In notation above
  $$
  \big(A_{j_1}A_{j_2}\dots A_{j_\ell}\big)_{ik} = |T^+_{ik}|-|T^-_{ik}|.
  $$
\end{prop}
\begin{proof}
  Restricting words in the set  $T^+_{ik}$ to the upper three rows one
  obtains all correct sequences $\tau$ such that
  $\eps(\tau)=1$. The same restriction for the set $T^-_{ik}$ gives
  the sequences $\tau$ such that $\eps(\tau)=-1$.

  The multiplicity of a sequence $\tau$ in the set
  $T^\pm_{ik}$ depends on fourth rows of the tables.
  According to the requirement (d) and
  the permutation property there are exactly $m^{j_\al}_{i_\al{}k_\al}$
  subwords  bearing 
   $(j_\al,i_\al,k_\al)$ in the upper three rows. So the multiplicity
  of the sequence $\tau$ is $m(\tau)$, where $m(\tau)$ is given
  by~\eqref{prod-element}.  
\end{proof}

It follows from the above construction that the sets $T^{\$}_{ik}(\boldsymbol
j)$, where $\$\in\{+,-,\text{bad}\}$,
do not intersect for different $\boldsymbol j$ because the sequence
$\boldsymbol j$ is recovered from the first row in the table
representation. 

\begin{prop} For a fixed collection of matrices  $A_1,\dots, A_N$ and
  for each pair $i,k$ the language
  \begin{equation*}
    \Ta{\$}_{ik}=\bigcup_{\boldsymbol j} T^{\$}_{ik}(\boldsymbol j)
  \end{equation*}
  is regular for  $\$\in\{+,-,\text{\upshape bad}\}$.
\end{prop}

\begin{proof}
  The requirements (a)--(d) are verified by local check on the
  subwords of the fixed length
  $\lceil\log_2 M\rceil$. 

  Computing the sign $\eps(\tau)$ can be done in the cyclic group of
  two elements. 

  So the sets $\Ta{\pm}_{ik}$ are regular. But the class of regular
  languages is closed under the complement. Thus the set
  $\Ta{\text{\textup{bad}}}_{ik}$ is also regular.
\end{proof}

\begin{proof}[Proof of Theorem~\textup{\ref{PePe:undecidable}}.]
We repeat the construction from Section~\ref{HCP}. 
The reduction algorithm
takes an instance  of the ZURC problem ($N=2$, $D=18$) and
constructs an automaton $C$ such that the
condition~\eqref{zero-expression}
holds for some element in the
semigroup generated by the input matrices if and only if $L(C)\cap
\pepe\ne\varnothing$.  

The automaton expects an input $w$ from the language $\pepe$ such that
block rank is $\ell\cdot\lceil\log_2 M\rceil$ and the word $w$ is a
certificate for zero representation~\eqref{zero-expression}. 
A period in the first
(periodic) component determines the sequence $\boldsymbol j=j_1,\dots,
j_\ell$ such that~\eqref{zero-expression} holds. The automaton expects
that each symbol in the sequence $\boldsymbol j$ is repeated
$\lceil\log_2 M\rceil$ times. In the second (permutation) component
the automaton expects that the  blocks from the sets
$T^+_{1D}(\boldsymbol j)$, $T^-_{1D}(\boldsymbol j)$ are paired and are followed by
the blocks from the set   $\Tb_{1D}$. 
Note that such a pairing exists
iff $|T^+_{1D}|=|T^-_{1D}|$. 
The structure of this part of
the automaton $C$ is similar to the automaton $C'$ described in
Section~\ref{HCP} and shown in Fig.~\ref{Auto=}.

The correctness of the reduction is proved in a way similar to the
arguments in Section~\ref{HCP}. If the automaton accepts a word $w$ in
the language $\pepe$ then one can extract the sequence $\boldsymbol j$
from the first components of symbols in the word~$w$. The check in the
second components guarantees that~\eqref{zero-expression} holds for the
sequence $\boldsymbol j$. We apply here
Proposition~\ref{pos-expansion}. 

In other direction, if~\eqref{zero-expression} holds for a
sequence $\boldsymbol j$ then there exists a word in the language
$\pepe$ satisfying the properties expected (and verified) by the
automaton~$C$. Thus  $L(C)\cap\pepe\ne\varnothing$.  
\end{proof}

\begin{remark}
  By suitable encoding 
of
the symbols of the alphabets $\Sigma_1$ and
  $\Sigma_2$ one can prove that the
  $(\Per_{\BB}{\parallel}P_{\BB})$-realizability problem is also
  undecidable.
\end{remark}

\subsection*{Acknowledgments}

Authors are grateful to J. Shallit and I. Sparlinski for helpful remarks.


\begin{thebibliography}{13}
  \bibitem{BellPotapov}
    \emph{Bell, P., Potapov I.} Lowering undecidability bounds for decision
    questions in matrices. In Developments in Language Theory.
    Springer Lecture Notes in Computer Science, vol. 4036, 2006. P.~375--385.

  \bibitem{BeRe} \emph{Berstel J.,  Reutenauer Ch.}
    Rational series and their languages.
    Springer-Verlag, 1988. 

    \bibitem{BP02} \emph{Blondel V. D., Portier N.} 
      The presence of a zero in an
    integer linear recurrent sequence is NP-hard to decide. Linear
    Algebra and Its Applications. 2002. Vol.~351--352. P.~91--98.

\bibitem{CookW}
\emph{Cook W., Fonlupt J., Schrijver A.}
An integer analogue of  Carath\'eodory's theorem. Journal of
Combinatorial Theory. Series B. 1986. Vol.~40, no~1. P.~63--70.


\bibitem{EisenShmonin} \emph{Eisenbrand F.,  Shmonin G.}
Carath\'eodory bounds for integer cones. Operations
Research Letters. 2006. Vol.~34. P.~564--568.

    \bibitem{GJ} \emph{Garey M. R., Johnson D. S.}
Computers and intractability: a guide to the theory of NP-completeness. San Francisco: Freeman,
 1979.


  \bibitem{Halava05} \emph{Halava V.,  Harju T.,  Hirvensalo M.,
    Karhum\"aki~J.} 
    Skolem's Problem---On the Border Between Decidability and
    Undecidability. TUCS Tech. Rep. No 683. 2005.

\bibitem{HopcroftMotwaniUllman} \emph{Hopcroft J., Motwani R., and
Ullman J.}  Introduction to automata theory, languages, and
computation. Boston: Addison-Wesley Company, 2001.

  \bibitem{LT09}  \emph{Laohakosol V.,  Tangsupphathawat P.}
  Positivity of third order linear recurrence sequences /\!/
  Discrete Applied Mathematics. 2009. 
  Vol.~157. Issue 15. P.~3239--3248. 

  \bibitem{Parikh}
  \emph{Parikh R. J.} On context-free languages. 1966. Journal of the
  ACM. Vol.~13, no 4. P.~570--581.

\bibitem{SaSo}
  \emph{Salomaa A., Soittola M.} Automata-theoretic aspects of formal power
  series.  Springer-Verlag, 1978.

\bibitem{Schrijver} \emph{Schrijver A.}
Theory of linear and integer programming. Chichester: John Wiley\& Sons, 1986.   

\bibitem{ShV} \emph{Shen A., Vereshchagin N. K.}
Computable functions. AMS, Providence, RI, 2003.

    \bibitem{Stanley} \emph{Stanley R.}
Enumerative combinatorics. Vol. 1.  Monterey: Wadsworth \& Brooks/Cole, 1986.

      \bibitem{Tao} \emph{Tao T.} Open question: effective
      Skolem-Mahler-Lech theorem.
\\
http://terrytao.wordpress.com/2007/05/25/open-question-effective-skolem-\\
mahler-lech-theorem/

\bibitem{VyaTar10} \emph{Vyalyi M., Tarasov S.}
Orbits of linear maps and regular languages. Discrete 
analysis and operations research. 2010. Vol. 17, No 6. P. 20--49 (in Russian).


  \bibitem{Ver} \emph{Vereshchagin N. K.}
On occurrence of zero in a linear recurrent sequence
/\!/
 Math. notes.
1985. v.~38, No 2. p.~177--189. 

\bibitem{Vya09}
\emph{Vyalyi M.N.} On models of a nondeterministic computation. Proc. of CSR
2009. Springer Lecture Notes in Computer Science, vol. 5675,
2009. P. 334--345.  

\end{thebibliography}
\end{document}